\documentclass[sigconf]{acmart}

\usepackage[ruled]{algorithm2e}
\usepackage{subfigure}

\renewcommand{\epsilon}{\varepsilon}
\newcommand{\R}{\mathbb{R}}
\newcommand{\Rnonneg}{\R_+}
\newcommand{\Xunit}{\Rnonneg}
\newcommand{\X}{\Xunit^d}
\newcommand{\M}{\mathcal{M}}
\newcommand{\E}{\mathbb{E}}
\newcommand{\N}{\mathbb{N}}
\newcommand{\RRound}[1]{\text{RandRound}\left(#1\right)}
\newcommand{\RResponse}[1]{\text{RandResponse}\left(#1\right)}
\newcommand{\range}{\mathcal{R}}
\newcommand{\ceil}[1]{{\left \lceil #1 \right \rceil}}
\newcommand{\floor}[1]{{\left \lfloor #1 \right \rfloor}}
\newcommand{\Lap}[1]{\mathrm{Lap}{\left(#1\right)}}
\newcommand{\Oovereps}[1]{O{\left(#1/\epsilon\right)}}
\newcommand{\Ooneeps}{\Oovereps{1}}
\newcommand{\Ologdeps}{O{\left(\frac{\log(d)}{\epsilon}\right)}}
\newcommand{\Ologdeltaeps}{O{\left(\frac{\log(1/\delta)}{\epsilon}\right)}}
\newcommand{\thresholdapprox}{Threshold2} 
\newcommand{\threshold}{Threshold} 
\newcommand{\ALPone}{ALP1-projection}
\newcommand{\ALPest}{ALP1-estimator}
\newcommand{\ALPgeneral}{ALP-projection}
\newcommand{\ALPgeneralest}{ALP-estimator} 
\newcommand{\ALPfinal}{Threshold ALP-projection}
\newcommand{\ALPfinalest}{Threshold ALP-estimator}

\renewcommand\footnotetextcopyrightpermission[1]{}

\AtBeginDocument{%
  \providecommand\BibTeX{{%
    \normalfont B\kern-0.5em{\scshape i\kern-0.25em b}\kern-0.8em\TeX}}}

\acmYear{2021}
\setcopyright{acmcopyright}
\acmConference[CCS '21] {Proceedings of the 2021 ACM SIGSAC Conference on Computer and Communications Security}{November 15--19, 2021}{Virtual Event, Republic of Korea.}
\acmBooktitle{Proceedings of the 2021 ACM SIGSAC Conference on Computer and Communications Security (CCS '21), November 15--19, 2021, Virtual Event, Republic of Korea}
\acmPrice{15.00}
\acmISBN{978-1-4503-8454-4/21/11} 
\acmDOI{10.1145/3460120.3484735}

\begin{document}

\title[Differentially Private Sparse Vectors with Low Error, Optimal Space, and Fast Access]{Differentially Private Sparse Vectors with Low Error, \\ Optimal Space, and Fast Access}  

\author{Martin Aumüller}
\email{maau@itu.dk}
\affiliation{%
  \institution{IT University of Copenhagen}
  \streetaddress{Rued Langgaards Vej 7}
  \city{Copenhagen}
  \country{Denmark}
  \postcode{2300}
}

\author{Christian Janos Lebeda}
\email{chle@itu.dk}
\affiliation{%
  \institution{BARC \\IT University of Copenhagen}
  \streetaddress{Rued Langgaards Vej 7}
  \city{Copenhagen}
  \country{Denmark}
  \postcode{2300}
}

\author{Rasmus Pagh}
\email{pagh@di.ku.dk} 
\affiliation{%
  \institution{BARC \\ University of Copenhagen}
  \streetaddress{Universitetsparken 1}
  \city{Copenhagen}
  \country{Denmark}
  \postcode{2100}
}

\fancyhead{}

\begin{abstract}
  Representing a sparse histogram, or more generally a sparse vector, is a fundamental task in differential privacy.
  An ideal solution would use space close to information-theoretical lower bounds, have an error distribution that depends optimally on the desired privacy level, and allow fast random access to entries in the vector.
  However, existing approaches have only achieved two of these three goals.
  
  In this paper we introduce the Approximate Laplace Projection (ALP) mechanism for approximating $k$-sparse vectors. This mechanism is shown to simultaneously have information-theoretically optimal space (up to constant factors), fast access to vector entries, and error of the same magnitude as the Laplace-mechanism applied to dense vectors.
  A key new technique is a \emph{unary} representation of small integers, which we show to be robust against ``randomized response'' noise. This representation is combined with hashing, in the spirit of Bloom filters, to obtain a space-efficient, differentially private representation.

  Our theoretical performance bounds are complemented by simulations which show that the constant factors on the main performance parameters are quite small, suggesting practicality of the technique.
\end{abstract}



\maketitle

\section*{Publication}

The paper appears in Proceedings of the 28th ACM Conference on Computer and Communications Security, 2021. 
\\
\url{https://doi.org/10.1145/3460120.3484735}

\section{Introduction}
\label{sec:intro}
One of the fundamental results in differential privacy is that a histogram can be made differentially private by adding noise from the Laplace distribution to each entry of the histogram before it is released~\cite{DworkDPpaper}.
The expected magnitude of the noise on each histogram entry is $\Ooneeps$, where $\epsilon$ is the privacy parameter, and this is known to be optimal~\cite{hardt2010geometry}. In fact, there is a sense in which the Laplace mechanism is optimal~\cite{koufogiannis2015optimality}. However, some histograms of interest are extremely sparse, and cannot be represented in explicit form. Consider, for example, a histogram of the number of HTTP requests to various servers. Already the IPv4 address space has over 4 billion addresses, and the number of unique, valid URLs have long exceeded $10^{12}$, so it is clearly not feasible to create a histogram with a (noisy) counter for each possible value.

Korolova, Kenthapadi, Mishra, and Ntoulas~\cite{korolova2009releasing} showed that it is possible to achieve \emph{approximate} differential privacy with space that depends only on the number of non-zero entries in the histogram. However, for $(\epsilon,\delta)$-differential privacy the upper bound on the expected per-entry error becomes $\Ologdeltaeps$, which is significantly worse than the Laplace mechanism for small $\delta$.
Cormode, Procopiuc, Srivastava, and Tran~\cite{cormode2012differentially} showed how to achieve pure $\epsilon$-differential privacy with expected per-entry error bounded by $\Ologdeps$, where $d$ is the dimension of the histogram, i.e., the number of entries.
While both these methods sacrifice accuracy they are very fast, allowing access to entries of the private histogram in constant time.
If access time is not of concern, it is possible to combine small space with small per-entry error, as shown by Balcer and Vadhan~\cite{balcer2017differential}. They achieve an error distribution that is comparable to the Laplace mechanism (up to constant factors) and space proportional to the sum $n$ of all histogram entries --- but the time to access a single entry is $\tilde{O}(n/\epsilon)$, which is excessive for large datasets.

\subsection{Our results}
Our contribution is a mechanism that  achieves optimal error and space (up to constant factors) with only a small increase in access time. The mechanism works for either approximate or pure differential privacy, with the  former providing
faster access time. Our main results are summarized in Theorem~\ref{thm:intro}.

\begin{theorem} [Informal Version of Theorems~\ref{thm:ALPpure} and ~\ref{thm:ALPapprox}]\label{thm:intro}
  Let $x$ be a histogram with $d$ entries each bounded by some value $u$ where at most $k$ entries have non-zero values. Given privacy parameters $\varepsilon > 0$ and $\delta \geq 0$, there exists an $(\varepsilon, \delta)$-differentially private algorithm to represent $x$ using $O(k\log(d+u))$ bits with per-entry error matching the Laplace mechanism up to constant factors. The access time is $O(\log(1/\delta))$ when $\delta > 0$ and $O(\log(d))$ when $\delta=0$.
\end{theorem}

Here we assume that $k=\Omega(\log(d))$. Otherwise the mechanism has an additional term of $O(\log^2(d))$ or $O(\log(d)\log(1/\delta))$ bits in its space usage for pure and approximate differential privacy, respectively.

\subsection{Techniques}
On a high level, we treat ``small'' and ``large'' values of the histogram differently. Large values are handled by the thresholding technique developed in~\cite{korolova2009releasing,cormode2012differentially}. For small entries,
we represent them using a \emph{unary encoding} as fixed-length bit strings.
From~\cite{korolova2009releasing,cormode2012differentially} we know that their length is logarithmic in either $d$ (for $\varepsilon$-DP) or $1/\delta$ (for $(\varepsilon, \delta$)-DP).
Privacy is achieved by perturbing each bit using randomized response~\cite{warner1965randomized}.
As it turns out, the unary encoding is redundant enough to allow accurate estimation even when the probability of flipping each bit is a constant bounded away from 1/2.
In order to pack all unary representations into small space, we use hashing to randomize the position of each bit in the unary representation of a given entry. 
The access time is linear in the length of the bit representation, given constant time evaluation of the hash function.
Interestingly, although hash collisions can lead to overestimates, they do not influence the error asymptotically. 

We remark here that a direct application of randomized response does not give the desired $\Ooneeps$ error dependency, but we solve this issue with an initial scaling step that gives $\epsilon$-differential privacy when combined with randomized response.
Though the discussion above has been phrased in terms of histograms, which makes the comparison to earlier work easier, our techniques apply more generally to representing sparse real vectors, with privacy for neighboring datasets with bounded $\ell_1$-distance.

\subsection{Overview}
In Section~\ref{sec:prelim} we define differential privacy for vectors, discuss the Laplace mechanism, and provide probabilistic tools necessary for the analysis. In Section~\ref{sec:relatedwork} we discuss related work on differentially private sparse histograms. In Section~\ref{sec:algorithm} we introduce the Approximate Laplace Projection (ALP) mechanism and analyze its theoretical guarantees. In Section~\ref{sec:combineddatastructure} we improve space and access time using techniques from earlier work~\cite{korolova2009releasing,cormode2012differentially}. In Section~\ref{sec:experiments} we evaluate the performance of the ALP mechanism based on simulations. In Section~\ref{sec:practioners} we present suggestions for practical applications. We conclude the paper by stating an open problem in Section~\ref{sec:openproblems}.

\section{Preliminaries}
\label{sec:prelim}
\paragraph{Problem Setup.} In this work, we consider $d$-dimensional $k$-sparse vectors of non-negative real values. We say that a vector $x \in \X$ is $k$-sparse if it contains at most $k$ non-zero entries. We assume that $k=\Omega(\log(d))$. All entries are bounded from above by a value $u \in \R$, i.e., $\max_{i \in [d]} x_i =: \|x\|_\infty \leq u$. Here $[d]$ is the set of integers $\{1, \dots, d\}$.
We consider the problem of constructing an algorithm $\M$ for releasing a differentially private representation of $x$, i.e., $\tilde{x}:= \M(x)$. Note that $\tilde{x}$ does not itself need to be $k$-sparse.

\paragraph{Utility Measures.} We use two measures for the utility of an algorithm $\M$. We define the \emph{per-entry error} as $|x_i - \tilde{x}_i|$ for any ${i \in [d]}$.
We define the \emph{maximum error} as $\max_{i \in [d]}|x_i-\tilde{x}_i|=\|x-\tilde{x}\|_\infty$. We compare the utility of algorithms using the expected per-entry and maximum error and compare the tail probabilities of the per-entry error of our algorithm with the Laplace mechanism introduced below.

\paragraph{Differential Privacy.} Differential privacy is a constraint to limit privacy loss introduced by Dwork, McSherry, Nissim, and Smith~\cite{DworkDPpaper}. 
We use definitions and results as presented by
Dwork and Roth~\cite{DworkBook}.
Intuitively, a differentially private algorithm ensures that a slight change in the input does not significantly impact the probability of seeing any particular output. We measure the distance between inputs using their $\ell_1$-distance. In this work, two vectors are neighbors iff their $\ell_1$-distance is at most 1. That is for all neighboring vectors $x, x' \in \X$ we have $\|x-x'\|_1 := \sum_{i \in [d]} |x_i - x'_i| \leq 1$. We can now define differential privacy for neighboring vectors.

\begin{definition}[Differential privacy~{\cite[Def~2.4]{DworkBook}}] \label{def:differentialprivacy}
    Given $\epsilon > 0$ and $\delta \geq 0$, a randomized algorithm $\M\colon \X \rightarrow \range$ is $(\epsilon, \delta)$-differentially private if for all subsets of outputs $S \subseteq \range$ and pairs of $k$-sparse input vectors $x, x' \in \X$ such that $\|x - x'\|_1 \leq 1$ it holds that:
    $$\Pr[\M(x) \in S] \leq e^\epsilon \cdot \Pr[\M(x') \in S] + \delta \enspace .$$
\end{definition}

$\M$ satisfies \emph{approximate differential privacy} when $\delta > 0$ and \emph{pure differential privacy} when $\delta = 0$. In particular, a pure differentially private algorithm satisfies $\epsilon$-\emph{differential privacy}. 
The following properties of differential privacy are useful in this paper.

\begin{lemma}[Post-processing~{\cite[Proposition~2.1]{DworkBook}}] \label{lem:postprocessing}
    Let $\M\colon \X \rightarrow \range$ be an $(\epsilon, \delta)$-differentially private algorithm and let $f\colon\range \rightarrow \range'$ be any randomized mapping. Then $f \circ \M\colon \X \rightarrow \range'$ is $(\epsilon, \delta)$-differentially private.   
\end{lemma} 

\begin{lemma}[Composition~{\cite[Theorem~3.16]{DworkBook}}] \label{lem:composition}
    Let $\M_1\colon \X \rightarrow \range_1$ and $\M_2\colon \X \rightarrow \range_2$ be randomized algorithms such that $M_1$ is $(\epsilon_1, \delta_1)$-differentially private and $\M_2$ is $(\epsilon_2, \delta_2)$-differentially private. Then the algorithm $\M$ where $\M(x)=(\M_1(x), \M_2(x))$ is $(\epsilon_1 + \epsilon_2, \delta_1 + \delta_2)$-differentially private.
\end{lemma} 

Throughout this paper, we clamp the output of all algorithms to the interval $\left[ 0,u \right]$. An estimate outside this interval is due to noise and clamping outputs cannot increase the error. It follows from Lemma~\ref{lem:postprocessing} that clamping the output does not affect privacy. We clamp the output implicitly to simplify presentation.

\paragraph{Probabilistic Tools.}
\emph{The Laplace Mechanism} introduced by Dwork, McSherry, Nissim, and Smith~\cite{DworkDPpaper} satisfies pure differential privacy by adding noise calibrated to the $\ell_1$-distance to each entry.
For completeness, Algorithm~\ref{alg:lapmech} provides a formulation of the Laplace mechanism in the context of releasing an $\varepsilon$-differentially private representation of a sparse vector.

\SetNlSty{textbf}{(}{)} %
\begin{algorithm}%
  \caption{The Laplace Mechanism \label{alg:lapmech}}
  \SetKwInOut{Parameters}{Parameters}\SetKwInOut{Input}{Input}\SetKwInOut{Output}{Output}
  \SetAlgoLined
  
  \Parameters{$\epsilon > 0$.}
  \Input{$k$-sparse vector $x \in \X$.}
  \Output{$\epsilon$-differentially private approximation of $x$.}
  
  \nl Let $\tilde{x}_i = x_i + \eta_i$ for all $i \in [d]$, where $\eta_i \sim \Lap{1/\epsilon}$. \\
  \nl Release $\tilde{x}$.
\end{algorithm} 

Here $\Lap{1/\epsilon}$ is the Laplace distribution with scale parameter $1/\epsilon$. The PDF and CDF of the distribution are presented in Definitions~\ref{def:lappdf} and~\ref{def:lapcdf} and the expected error and tail bound of the mechanism are shown in Propositions~\ref{prop:laperror} and~\ref{prop:laptail}. 
The Laplace mechanism works well for vectors with low dimensionality and serves as a baseline for our work. However, it is impractical or even infeasible in the setting of $k$-sparse vectors. The output vector is dense, and as such the space requirement scales linearly in the input dimensionality $d$.

\begin{definition} \label{def:lappdf}
  The probability density function of the Laplace distribution centered around 0 with scale parameter $1/\epsilon$ is
  $$f(\tau) = \frac{\epsilon}{2}e^{-|\tau|\epsilon} \enspace .$$
\end{definition}

\begin{definition} \label{def:lapcdf}
  The cumulative distribution function of the Laplace distribution centered around 0 with scale parameter $1/\epsilon$ is:
  $$\Pr[\Lap{1/\epsilon} \leq \tau] = \begin{cases}
    \frac{1}{2}e^{\tau\epsilon}, & \text{if } \tau < 0 \\
    1 - \frac{1}{2}e^{-\tau\epsilon}, & \text{if } \tau \geq 0
  \end{cases}$$
\end{definition}

\begin{proposition} [Expected Error~{\cite[Theorem~3.8]{DworkBook}}] \label{prop:laperror}
  The expected per-entry and maximum error of the Laplace mechanism are $\E[|x_i-\tilde{x}_i|] = \Ooneeps$ and $\E[\|x-\tilde{x}\|_\infty] = \Ologdeps$ respectively.
\end{proposition}

\begin{proposition} [Tail bound~{\cite[Theorem~3.8]{DworkBook}}]\label{prop:laptail}
  With probability at least $1-\psi$ we have:
  $$|\Lap{1/\epsilon}| \leq \frac{1}{\epsilon}\ln\frac{1}{\psi} \enspace .$$
\end{proposition}

\emph{Random rounding} or stochastic rounding is used for rounding a real value probabilistically based on its fractional part. We define random rounding for any real $r \in \R$ as follows:
$$\RRound{r} = \begin{cases}
  \ceil{r} & \text{with probability } r - \floor{r} \\
  \floor{r} & \text{with probability } 1 - (r - \floor{r}) 
\end{cases}$$

\begin{lemma} \label{lem:rrounderror}
  The expected error of random rounding is maximized when $r-\floor{r}=0.5$. For any $r$ we have:
  \begin{align*}
    \E[|r - \RRound{r}|]\leq \frac{1}{2} \enspace .
  \end{align*}
\end{lemma}

\emph{Randomized response} was first introduced by Warner~\cite{warner1965randomized}. The purpose of the mechanism is to achieve plausible deniability by changing one's answer to some question with probability $p$ and answer truthfully with probability $q=1-p$. We define randomized response for a boolean value $b \in \{0, 1\}$ as follows: 
$$\RResponse{b, p} = \begin{cases}
  1-b & \text{with probability } p \\
  b & \text{with probability } q 
\end{cases}$$

\paragraph{Universal Hashing.} A \emph{hash family} is a collection of functions $\mathcal{H}$ mapping keys from a universe $U$ to a range $R$. A family $\mathcal{H}$ is called \emph{universal}, if each pair
of different keys collides with probability at most $1/|R|$,
where the randomness is taken over the random choice of $h \in \mathcal{H}$. 
A particularly efficient construction that uses $O(\log |U|)$ bits and constant evaluation time is presented in~\cite{DietzfelbingerHKP97}.

\paragraph{Model of Computation.} We use the $w$-bit word RAM model defined by Hagerup~\cite{Hagerup98} where $w=\Theta(\log(d) + \log(u))$. This model allows constant time memory access and basic operations on $w$-bit words. As such, we can store a $k$-sparse vector using $O(k\log(d+u))$ bits with constant lookup time using a hash table. 
We assume that the privacy parameters $\varepsilon$ and $\delta$ can be represented in a single word.

\paragraph{Negative Values.} In this paper, we consider vectors with non-negative real values, but the mechanism can be generalized for negative values using the following reduction. Let $v \in \R^d$ be a real valued $k$-sparse vector. Construct $x, y \in \R_+^d$ from $v$ such that $x_i = \max(v_i, 0)$ and $y_i = -\min(v_i,0)$. By construction both $x$ and $y$ are $k$-sparse and the $\ell_1$-distance between vectors is preserved. We can access elements in $v$ as $v_i=x_i-y_i$. As such, any differentially private representation of $x$ and $y$ can be used as a differentially private representation of $v$ with at most twice the error.

\section{Related work}
\label{sec:relatedwork}

Previous work on releasing differentially private sparse vectors primarily focused on the special case of discrete vectors in the context of releasing the histogram of a dataset. 

\begin{table*}[t]
  \centering
  \def\arraystretch{1.5}
  \begin{tabular}{lcccc}
  \multicolumn{1}{l}{\textbf{Algorithm}}                & \multicolumn{1}{c}{\textbf{Space (bits)}}                     & \multicolumn{1}{c}{\textbf{Access time}}            & \multicolumn{1}{c}{\textbf{Per-entry error}}                                                                         & \multicolumn{1}{c}{\textbf{Maximum error}}     \\ \hline
  Dwork et al. \cite{DworkDPpaper}                       & $O{\left(d\log(u)\right)}$                            & $O{\left(1\right)}$                          & $O{\left(\frac{1}{\epsilon}\right)}$   & $\Ologdeps$ \\ 
  Cormode et al. \cite{cormode2012differentially}        & $O{\left(k \log(d + u)\right)}$                       & $O{\left(1\right)}$                          & $\Ologdeps$                            & $\Ologdeps$ \\ 
  Balcer \& Vadhan \cite{balcer2017differential}         & $\tilde{O}{\left(\frac{n}{\epsilon} \log(d) \right)}$ & $\tilde{O}{\left(\frac{n}{\epsilon}\right)}$ & $O{\left(\frac{1}{\epsilon}\right)}$   & $\Ologdeps$ \\ 
  \textit{Theorem~\ref{thm:ALPpure} (this work)}                  & $O{\left(k \log(d + u) \right)}$                      & $O{\left(\log(d)\right)}$                    & $O{\left(\frac{1}{\epsilon}\right)}$   & $\Ologdeps$ \\ \hline 
  Korolova et al. \cite{korolova2009releasing} & $O{\left(k \log(d + u) \right)}$                      & $O{\left(1\right)}$                          & $\Ologdeltaeps$                        & $\Ologdeltaeps$ \\ 
  \textit{Theorem~\ref{thm:ALPapprox} (this work)}      & $O{\left(k (\log(d + u) + \log(1/\delta) \right))}$           & $O{\left(\log(1/\delta)\right)}$             & $O{\left(\frac{1}{\epsilon}\right)}$   & $\Ologdeltaeps$ \\ \hline
\end{tabular}%
  \captionsetup{justification=centering}
  \caption{Comparison with previous work of expected values for worst-case input. 
  The first four rows are results on $\epsilon$-differential privacy, and the last two are on $(\epsilon,\delta)$-differential privacy. The $\tilde{O}$-notation suppresses logarithmic factors.
  }
  \label{tab:relatedworkcomparison}
\end{table*} 
Korolova, Kenthapadi, Mishra, and Ntoulas~\cite{korolova2009releasing} first introduced an approximately differentially private mechanism for the release of a sparse histogram. A similar mechanism was later introduced independently by Bun, Nissim, and Stemmer~\cite{bun2016simultaneous} in another context. 
The mechanism adds noise to non-zero entries and removes those with a noisy value below a threshold $t=\Ologdeltaeps$. The threshold is chosen such that the probability of releasing an entry with true value $1$ is at most $\delta$. The expected maximum error is $O{\left(\frac{\log(\max(k,1/\delta))}{\epsilon}\right)}$. Since $\delta$ is usually chosen to be negligible in the input size, we assume that $\delta \leq 1/k$. As such, the expected maximum error is $\Ologdeltaeps$. 
We discuss the per-entry error below. 
Their mechanism is designed to satisfy differential privacy for discrete data. We extend their technique to real-valued data as part of Section~\ref{sec:combineddatastructure}, where we combine it with our mechanism.

Cormode, Procopiuc, Srivastava, and Tran~\cite{cormode2012differentially} introduced a differentially private mechanism in their work on range queries for sparse data. The mechanism adds noise to all entries and removes those with a noisy value below a threshold $t=\Ologdeps$. Here the threshold is used to reduce the expected output size. The number of noisy entries above $t$ is $O(k)$ with high probability. 
The construction time of a naive implementation of their technique scales linearly in $d$. They improve on this by sampling from a binomial distribution to determine the number of zero entries to store. They show that their approach produces the same output distribution as a naive implementation that adds noise to every entry. 
Their mechanism works for real-valued data in a straightforward way.

Since the expected number of non-zero entries in the output is $O(k)$ for both mechanisms above, their memory requirement is $O(k\log(d+u))$ bits using a hash table. An entry is accessed in constant time. The expected per-entry error depends on the true value of the entry. If the noisy value is above the threshold with sufficiently high probability, the expected error is $\Ooneeps$. However, this does not hold for entries that are likely removed. Consider for example an entry with a true value exactly at the threshold $t$. This entry is removed for any negative noise added. As such the expected per-entry error is $O(t)$ for worst-case input, which is $\Ologdeltaeps$ and $\Ologdeps$ for the two mechanisms, respectively.

In their work on differential privacy on finite computers, Balcer and Vadhan~\cite{balcer2017differential} introduced several algorithms including some with similar utility as the mechanisms described above. Moreover, they provided a lower bound of $\Omega{\left(\frac{\min\{\log(d),\, \log(\varepsilon/\delta),\, n\}}{\epsilon}\right)}$ for the expected per-entry error of any algorithm that always outputs a sparse histogram. (See~\cite[Theorem 7.2]{balcer2017differential} for the precise technical statement.) Here $n$ is the number of rows in the dataset, i.e., the sum of all entries of the histogram. This lower bound means that an algorithm that always outputs a $O(k)$-sparse histogram cannot achieve $\Ooneeps$ expected per-entry error for all input. They bypass this bound by producing a compact representation of a dense histogram. Their representation has expected per-entry and maximum error of $\Ooneeps$ and $\Ologdeps$, respectively. It requires $\tilde{O}{\left(\frac{n}{\epsilon} \log(d)\right)}$ bits and an entry is accessed in time $\tilde{O}{\left(\frac{n}{\epsilon}\right)}$. Note that their problem setup differs from ours in that each entry is bounded only by $n$ such that $\|x\|_\infty \leq n$. That is, $n$ serves a similar purpose as $u$ does in our setup. We do not know how to extend their approach to our setup with real-valued input.

In light of the results achieved in previous work, our motivation is to design a mechanism that achieves three properties simultaneously: $\Ooneeps$ expected per-entry error for arbitrary input, fast access, and (asymptotically) optimal space. Previous approaches only achieved at most two of these properties simultaneously. Moreover, we want the per-entry error to match the tail bounds of the Laplace mechanism up to constant factors.
We construct a compact representation of a dense vector to bypass the lower bound for sparse vectors by Balcer and Vadhan~\cite{balcer2017differential}. The access time of our mechanism is $O(\log(d))$ and $O(\log(1/\delta))$ for pure and approximate differential privacy, respectively. Table~\ref{tab:relatedworkcomparison} summarizes the results of previous work and our approach. 

\section{The ALP mechanism}
\label{sec:algorithm}

In this section, we introduce the Approximate Laplace Projection (ALP) mechanism\footnote{The name is chosen to indicate that the error distribution is approximately like the Laplace distribution, and that we \emph{project} the sparse vector to a much lower-dimensional representation. It also celebrates the mountains, whose silhouette plays a role in a certain random walk considered in the analysis of the ALP mechanism.} and give an upper bound on the expected per-entry error. 
The ALP mechanism consists of two algorithms. The first algorithm constructs a differentially private representation of a $k$-sparse vector and the second estimates the value of an entry based on its representation.

\subsection{A 1-differentially private algorithm}

We start by considering the special case of $\epsilon=1$ and later generalize to all values of $\epsilon > 0$. 
Moreover, the mechanism works well only for entries bounded by a parameter $\beta$. 
In general, this would mean that we had to set $\beta=u$ if we only were to use the ALP mechanism.
However, in Section~\ref{sec:combineddatastructure} we will discuss how to set $\beta$ smaller and still perform well for all entries. 

In the first step of the projection algorithm, we scale every non-zero entry by a parameter of the algorithm and use random rounding to map each such entry to an integer.
We then store the unary representation of these integers in a two-dimensional bit-array using a sequence of universal hash functions~\cite{carter1979universal}. We call this bit-array the \emph{embedding}. Lastly, we apply randomized response on the embedding to achieve privacy. The pseudocode of the algorithm is given in Algorithm~\ref{alg:ALP1} and we discuss it next.

\begin{algorithm}[t]
  \caption{\ALPone \label{alg:ALP1}}
  \SetKwInOut{Parameters}{Parameters}\SetKwInOut{Input}{Input}\SetKwInOut{Output}{Output}
  \SetAlgoLined

  \Parameters{$\alpha, \beta > 0$, and $s \in \N$.}
  \Input{$k$-sparse vector $x \in \X$ where $s > 2k$. Sequence of hash functions from domain $[d]$ to $[s]$, $h = (h_1,\ldots,h_m)$, where $m=\ceil{\frac{\beta}{\alpha}}$.} 
  \Output{$1$-differentially private representation of $x$.}
  
  \nl Apply random rounding to a scaled version of each non-zero entry of $x$ such that $y_i = \RRound{\frac{x_i}{\alpha}}$. \\
  \nl Construct $z \in \{0, 1\}^{s \times m}$ by hashing the unary representations of $y$ such that:
  $$z_{a,b} = \begin{cases}
    1, & \exists i: b \leq y_i \text{ and } h_b(i) = a \\
    0, & \text{otherwise}
  \end{cases} $$ \label{line:ALPz} \\
  \nl Apply randomized response to each bit of $z$ such that $\tilde{z}_{a,b} = \RResponse{z_{a,b}, \frac{1}{\alpha+2}}$. \\
  \nl Release $h$ and $\tilde{z}$.
\end{algorithm} 

\begin{figure}[t]
  \includegraphics[width=\linewidth]{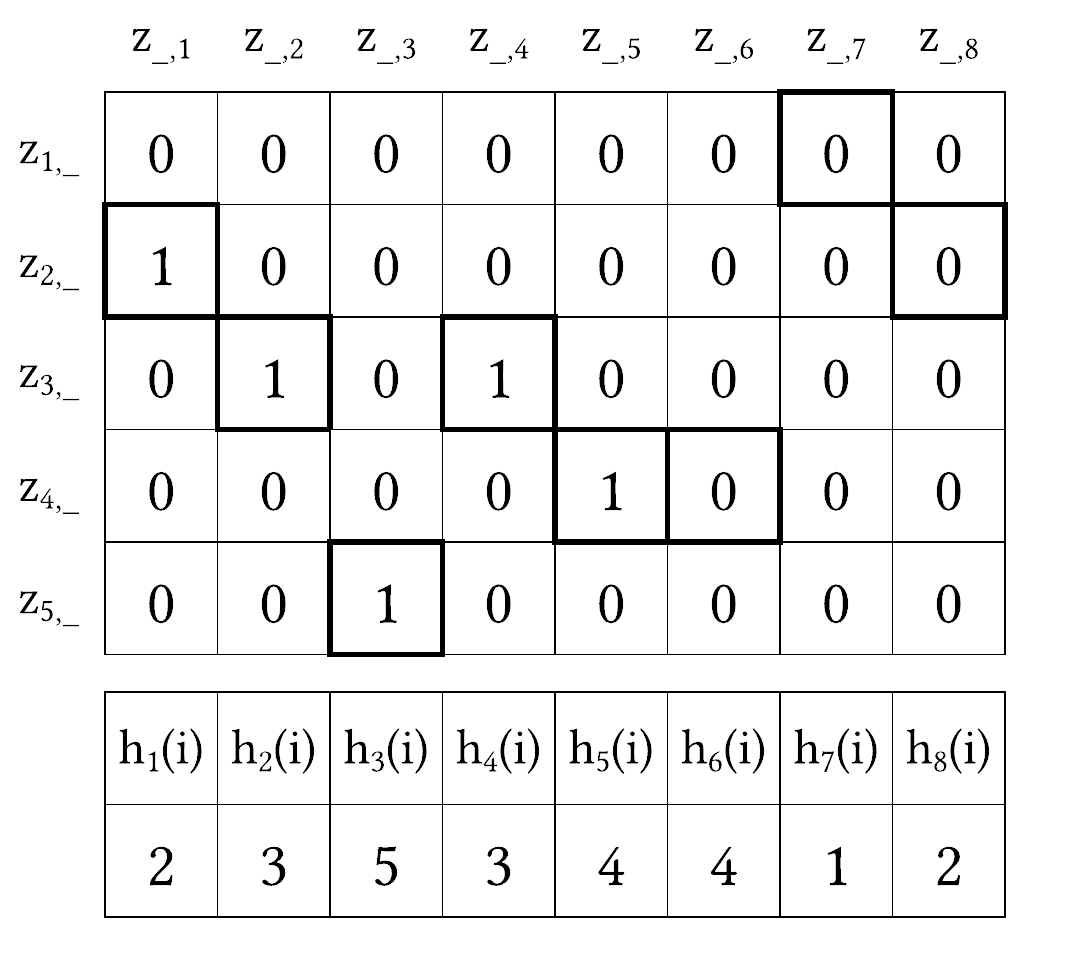} 
  \Description[Visualization of ALP-projection]{We use $m$ bits to represent each entry. Since we are storing the value 5, we set the first 5 bits to 1.}
  \captionsetup{justification=centering}
  \caption{Embedding with $m=8$, $s=5$ and $y_i=5$. \\
  The $i$th entry is the only non-zero entry.} 
  \label{fig:projection}
\end{figure}

Figure~\ref{fig:projection} shows an example of an embedding before applying randomized response. The input is a vector $x$ where the $i$th entry $x_i$ is the only non-zero value. The result of evaluating $i$ for each hash function is shown in the table at the bottom and the $m=8$ bits representing the $i$th entry in the bit-array are highlighted. 
In Step (1) of the algorithm, $x_i$ is scaled by $1/\alpha$ and randomized rounding is applied to the scaled value.
This results in $y_i = 5$. Using the hash functions, we represent this value in unary encoding by setting the first five bits to $1$ in Step (2),
where the $j$th bit is selected by evaluating the hash function $h_j$ on $i$.
The final three bits are unaffected by the entry. 
Finally, we apply randomized response in each cell of the bit-array. 
The bit-array after applying randomized response is not shown here, but we present it later in Figure~\ref{fig:estimation}.
Both the bit-array and the hash functions are the differentially private representation of the input vector $x$. We use this construction when estimating the value of $x_i$ later.

The algorithm takes three parameters $\alpha, \beta$, and $s$.
The parameters $\alpha$ and $s$ are adjustable.
We discuss these parameters later as part of the error analysis. In Section~\ref{sec:experiments} we discuss how to select values for $\alpha$ and $s$. Throughout the paper we sometimes assume that $\alpha$ is a constant and $s$ is a constant multiple of $k$ that is $\alpha=\Theta(1)$ and $s=\Theta(k)$. The parameter $\beta$ bounds the values stored in the embedding. We discuss $\beta$ as part of the error analysis as well.

\begin{lemma} \label{lem:ALP1DP}
  Algorithm~\ref{alg:ALP1} satisfies 1-differential privacy.
\end{lemma}
\begin{proof}
Let $x, x' \in \X$ denote two neighboring vectors. 
We prove the lemma in several steps. 
First, the vectors differ only in their $i$th entry. In this case, we start by assuming that only a single bit of $z$ is affected by changing $x$ to $x'$ and that there are no hash collisions. We then allow them to differ in several bits and include hash collisions. Finally, we generalize to the case that they differ in more than one entry.

Assume that $z$ differs only in a single bit for $x$ and $x'$. 
Let $Y$ denote the event that the affected bit is set to one after running the algorithm. Let $p=\frac{1}{\alpha+2}$ be the parameter of the randomized response step and let $q=1-p$. Then the probability of $Y$ occurring with input $x$ is $\Pr[Y \mid x] = (1-r) \cdot p + r \cdot q$, where $r=\frac{x_i}{\alpha} - \floor{\frac{\min(x_i, x'_i)}{\alpha}}$ denotes the probability of the bit being one before the randomized response step. 
Similarly for $x'$ we define $r'=\frac{x'_i}{\alpha} - \floor{\frac{\min(x_i, x'_i)}{\alpha}}$. The minimum term is needed when $\max(x_i, x'_i)$ is a multiple of $\alpha$ such that $\max(r,r')=1$. We find the difference in the probability of $Y$ occurring for $x$ and $x'$ as:
\begin{align*} 
  \Pr[Y \mid x] - \Pr[Y \mid x'] & = ((1-r)p + rq) - ((1-r')p + r'q) \\
  & = (r - r') \cdot (q-p) \\
  & = \frac{x_i-x'_i}{\alpha} \cdot \frac{\alpha}{\alpha+2} \\
  & = \frac{x_i-x'_i}{\alpha+2} \enspace .
\end{align*} 

By symmetry, the absolute difference in probability for setting the bit to either zero or one is $\frac{|x_i-x'_i|}{\alpha+2}$. Let $Z$ be an arbitrary output of Algorithm~\ref{alg:ALP1}. Since $x$ and $x'$ agree on all but the $i$th entry, the change in probability of outputting $Z$ depends only on the affected bit. 
Let $Y'$ denote the event that the bit agrees with output $Z$. Then we find the ratio of probabilities of outputting $Z$ as:

\begin{align*}
  \frac{\Pr[\text{\ALPone}(x') = Z]}{\Pr[\text{\ALPone}(x) = Z]} & = \frac{\Pr[Y' \mid x']}{\Pr[Y' \mid x]} \leq \frac{\Pr[Y' | x] + \frac{|x_i-x'_i|}{\alpha + 2}}{\Pr[Y' | x]} \\
  & \leq \frac{p + \frac{|x_i-x'_i|}{\alpha+2}}{p} 
  = 1 + |x_i - x'_i| \\
  & \leq e^{|x_i - x'_i|} \enspace .
\end{align*}

Here the second inequality follows from $p \leq \Pr[Y' | x] \leq q$. From here it is easy to take hash collisions into account as follows: Let $p'$ denote the probability of $Y$ occurring after setting the $i$th entry to zero. That is, we have $p \leq p' \leq q$ and $\Pr[Y \mid x]=(1-r)\cdot p' + r \cdot q$. The absolute difference in probability is still bounded such that $\Pr[Y \mid x] - \Pr[Y \mid x'] \leq \frac{|x_i-x'_i|}{\alpha+2}$. As such it still holds that:

\begin{align*}
  \frac{\Pr[\text{\ALPone}(x') = Z]}{\Pr[\text{\ALPone}(x) = Z]} & \leq e^{|x_i-x'_i|} \enspace .
\end{align*}

Next, we remove the assumption that only a single bit is affected by composing probabilities. We provide the following inductive construction. Let $x, x' \in \X$ be vectors that differ in the $i$th entry such that exactly two bits are affected. We consider the case of $x_i < x_i'$ and fix a vector $x'' \in \X$ with $x_i < x''_i < x'_i$ such that the differences affects exactly one bit each. 
Again, let $Z$ be an arbitrary output of Algorithm~\ref{alg:ALP1}.
Applying the upper bound from above twice, we may bound the change in probabilities by:

\begin{align*}
  \frac{\Pr[\text{\ALPone}(x') = Z]}{\Pr[\text{\ALPone}(x) = Z]} & = \frac{\Pr[\text{\ALPone}(x'') = Z]}{\Pr[\text{\ALPone}(x) = Z]} \\
  & \quad \cdot \frac{\Pr[\text{\ALPone}(x') = Z]}{\Pr[\text{\ALPone}(x'') = Z]} \\
  & \leq e^{|x_i - x''_i|} \cdot e^{|x''_i - x'_i|} \\
  & = e^{|x_i-x'_i|} \enspace ,
\end{align*}
which can be applied inductively if changing an entry affects more than two bits.

We are now ready to generalize to any vectors $x, x' \in \X$, i.e.,
where vectors may differ in more than a single position. Using the bound from above, we can bound the ratio of probabilities by: 
\begin{align*}
  \frac{\Pr[\text{\ALPone}(x') = Z]}{\Pr[\text{\ALPone}(x) = Z]} &\leq \prod_{i \in [d]} e^{|x_i-x'_i|} \\
  &= e^{\sum_{i\in [d]} |x_i-x'_i|} \\
  &= e^{\|x-x'\|_1} \enspace .
\end{align*}

The privacy loss is thus bounded by the $\ell_1$-distance of the vectors for any output. Recall that the $\ell_1$-distance is upper bounded by $1$ for two neighboring vectors.
As such the algorithm is $1$-differentially private as for any pair of neighboring vectors $x$ and $x'$ and any subset of outputs $S$ we have: 
\begin{align*}
  \Pr[\text{\ALPone}(x) \in S] & \leq e^{\|x-x'\|_1} \Pr[\text{\ALPone}(x') \in S] \\
  & \leq e \cdot \Pr[\text{\ALPone}(x') \in S]  \enspace .
\end{align*}
\end{proof}

The following lemma summarizes the space complexity of storing the bit-array and the collection of hash functions.
\begin{lemma} \label{lem:ALP1space}
  The number of bits required to store $h$ and $\tilde{z}$ is
  \begin{align*}
    O{\left(\frac{(s + \log d ) \cdot \beta}{\alpha}\right)} \enspace .
  \end{align*}
\end{lemma}

\begin{proof}
  By definition $m=O{\left(\frac{\beta}{\alpha}\right)}$ and as such $s\cdot m=O{\left(\frac{s\beta}{\alpha}\right)}$ bits are used to store $\tilde{z}$. Each hash function uses $O(\log (d))$ bits for a total of $O{\left(\frac{\log(d)\beta}{\alpha}\right)}$ bits to store $h$.
\end{proof}

\subsection{Estimating an entry}
We now introduce the algorithm to estimate an entry based on the embedding from Algorithm~\ref{alg:ALP1}. When accessing the $i$th entry, we estimate the value of $y_i$ and multiply by $\alpha$ to reverse the initial scaling of $x_i$. The estimate of $y_i$ is chosen to maximize a partial sum. If multiple values maximize the sum we use their average.
\paragraph{Intuition.} The first $y_i$ bits representing the $i$th entry are set to one before applying noise in Algorithm~\ref{alg:ALP1}, cf. Figure~\ref{fig:projection}. The last $m-y_i$ bits are zero, except if there are hash collisions. 
Some bits might be flipped due to randomized response, but we expect the majority of the first $y_i$ bits to be ones and the majority of the remaining $m-y_i$ bits to be zeros. As such the estimate of $y_i$ is based on prefixes maximizing the difference between ones and zeros. The pseudocode for the algorithm is given as Algorithm~\ref{alg:ALPestimator}.

\begin{algorithm}[t]
  \caption{\ALPest \label{alg:ALPestimator}}
  \SetKwInOut{Parameters}{Parameters}\SetKwInOut{Input}{Input}\SetKwInOut{Output}{Output}
  \SetAlgoLined

  \Parameters{$\alpha > 0$.}
  \Input{Embedding $\tilde{z} \in \{0,1\}^{s \times m}$. Sequence of hash functions $h=(h_1,\ldots,h_m)$. Index $i \in [d]$.}
  \Output{Estimate of $x_i$.}

  \nl Define the function $f\colon \{0,\dots,m\} \rightarrow \mathbb{Z}$ as:
  $$f(n)= \sum_{a=1}^n 2\tilde{z}_{h_a(i),a}-1$$ \\
  \nl Let $P$ be the set of arguments maximizing $f$. That is,
  $$P = \{ n \in \{0,\dots,m\} : f(a) \leq f(n) \text{ for all }a \in \{0,\dots,m\}\}$$ \\
  \nl Let $\tilde{y}_i = \mathrm{average}(P)$ \\
  \nl Return $\tilde{y}_i \cdot \alpha$.
\end{algorithm}

\begin{figure}[t]
  \includegraphics[width=\linewidth]{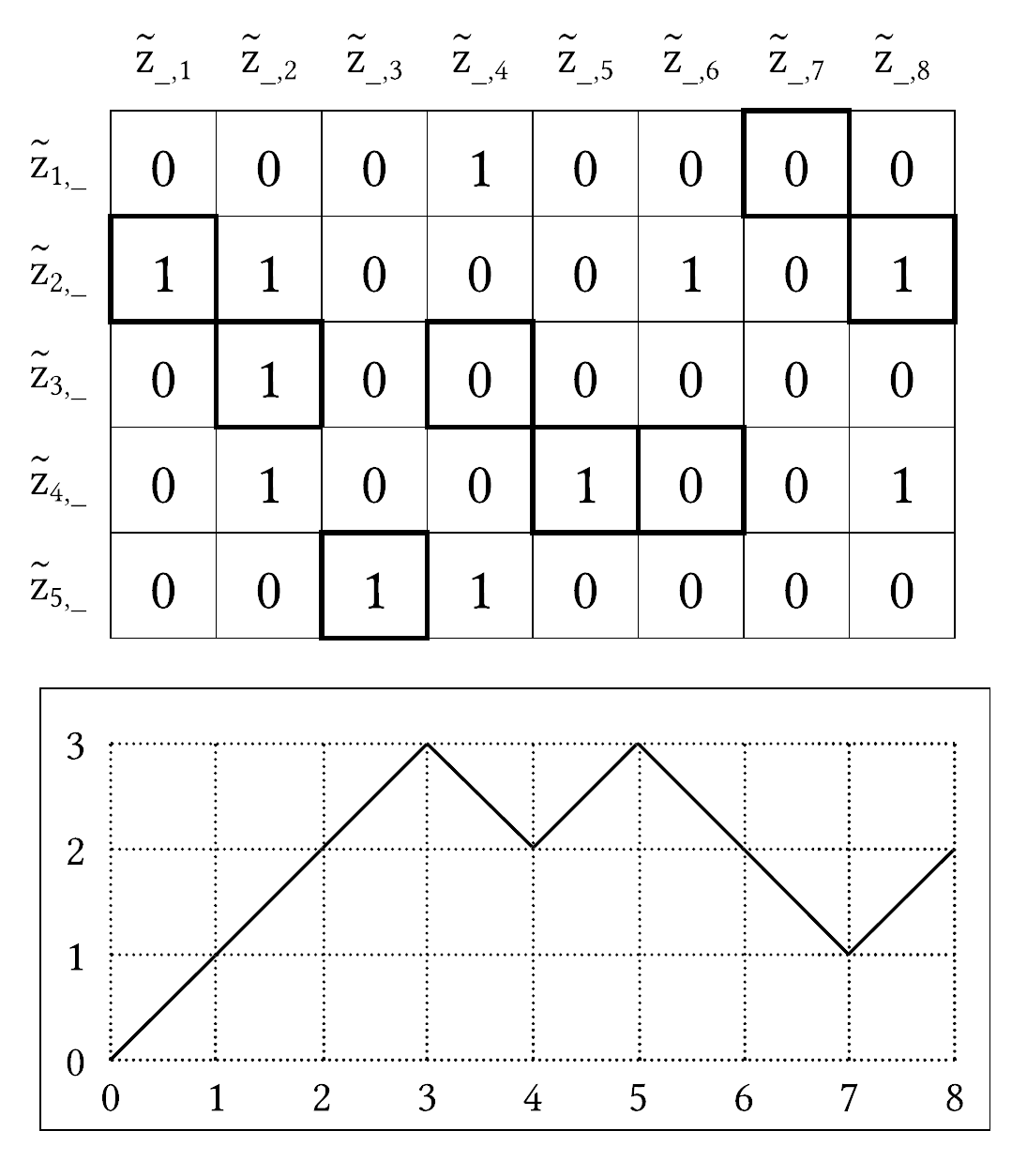} 
  \Description[Visualization of ALP1-estimation]{We use $m$ hash functions to estimate the projected entry.}
  \captionsetup{justification=centering}
  \caption{Estimation of $i$th entry from Figure~\ref{fig:projection}.\\
  The partial sum is maximized at indices 3 and 5. \\
  The estimate is 4, while the true value was 5.
  } 
  \label{fig:estimation}
\end{figure}

Figure~\ref{fig:estimation} shows an example of Algorithm~\ref{alg:ALPestimator}. The example is based on the embedding from Figure~\ref{fig:projection} after adding noise.
The plot shows the value of $f$ for all candidate estimates. 
This sum is maximized at positions $3$ and $5$. This is visualized as the global \emph{peaks} in the plot. The estimate is the average of those positions. 

\begin{lemma} \label{lem:ALP1estrunningtime}
  The evaluation time of Algorithm~\ref{alg:ALPestimator} is $O{\left(\frac{\beta}{\alpha}\right)}$. 
\end{lemma}

\begin{proof}
  We can compute all partial sums by evaluating each bit $(\tilde{z}_{h_1(i),1},\ldots,\tilde{z}_{h_m(i),m})$ once using dynamic programming. As such the evaluation time is $O(m)$ with $m=\ceil{\frac{\beta}{\alpha}}$. We have $m=O(\beta)$ when $\alpha=\Theta(1)$.
\end{proof}

We now analyze the per-entry error of Algorithm~\ref{alg:ALPestimator}. We first analyze the expected error based on the parameters of the algorithm. The results are presented in Lemma~\ref{lem:ALP1fixederror}. In Lemmas~\ref{lem:probboundonerror} and~\ref{lem:maxerrorwithprob} we bound the tail distribution of the per-entry error of the algorithm.

\begin{lemma} \label{lem:ALP1errorroughbound}
  The expected per-entry error of Algorithm~\ref{alg:ALPestimator} is bounded by $(\frac{1}{2} + \E[|y_i - \tilde{y}_i|]) \cdot \alpha$ for entries with a value of at most $\beta$.
\end{lemma}

\begin{proof}
  It is clear that the error of the $i$th entry is $\alpha$ times the difference between $\tilde{y}_i$ and $\frac{x_i}{\alpha}$. The expected difference is bounded by:
  \begin{align*}
    \E[|\frac{x_i}{\alpha}-\tilde{y}_i|] & \leq \E[|\frac{x_i}{\alpha}-y_i|] + \E[|y_i - \tilde{y}_i|] \\
    & \leq \frac{1}{2} + \E[|y_i - \tilde{y}_i|]  \enspace .
  \end{align*}
  The last inequality follows from Lemma~\ref{lem:rrounderror}.
\end{proof}

We find an upper bound on $\E[|y_i - \tilde{y}_i|]$ by analyzing simple random walks. A simple random walk is a stochastic process such that $S_0=0$ and $S_n=\sum_{\ell=1}^n X_\ell$, where $X$ are independent and identically distributed random variables with $\Pr[X_\ell = 1] = p$ and $\Pr[X_\ell = -1] = 1 - p = q$. 

\begin{lemma} \label{lem:randomwalkmonkey}
  Let $S$ be a simple random walk with $p<q$. At any step $n$ the probability that there exists a later step $\ell > n$ such that $S_\ell > S_n$ is $\frac{p}{q}$.
\end{lemma}

\begin{proof}
  It follows directly from Theorem~1 by Alm~\cite{alm2002simple}.
\end{proof}

For our analysis, we are concerned with the maximum $n$ such that $S_n \geq 0$. For an infinite random walk where $p < q$ such an $n$ exists with probability $1$. 

\begin{lemma} \label{lem:randomwalk}
  Let $S$ be a simple random walk with $p<q$. The expected last non-negative step of $S$ is: $\E[\max_n: S_n \geq 0]=\frac{4pq}{(q-p)^2}$.
\end{lemma}

\begin{proof}
  We use Lemma~\ref{lem:randomwalkmonkey} to find the probability that $S_n$ is the unique maximum in $\{S_n,\ldots,S_\infty\}$ as follows:
  \begin{align*}
    \Pr[S_n > \max(\{S_{n+1},\ldots,S_\infty\})] & = \Pr[X_{n+1} = -1] \cdot \\
    & \quad \Pr[S_{n+1} = \max(\{S_{n+1},\ldots,S_\infty\})] \\
    & = q \cdot (1-\frac{p}{q}) \\
    & = q - p  \enspace .
  \end{align*}
  
  The last non-negative step must have value exactly zero and as such must be at an even numbered step. The probability that step $2i$ is the last non-negative is:
  \begin{align*}
    \Pr[(\max_n: S_n \geq 0) = 2i] & = \Pr[S_{2i} = 0] \cdot \\
    & \quad \Pr[S_i > \max(\{S_{i+1},\ldots,S_\infty\})] \\
    & = {2i \choose i} (pq)^i (q-p) \enspace .
  \end{align*}

  We are now ready to find the expected last non-negative step of an infinite simple random walk as:
  \begin{align*}
    \E[\max_n: S_n \geq 0] & = \sum_{i=0}^\infty 2i \cdot \Pr[(\max_n: S_n \geq 0 ) = 2i] \\
    & = \sum_{i=0}^\infty 2i {2i \choose i} (pq)^i (q-p) \\
    & = 2(q-p) \sum_{i=0}^\infty i {2i \choose i} (pq)^i \\
    & = \frac{4pq}{(q-p)^2} \enspace .
  \end{align*}

  The last equality follows from the identity $\sum_{i=0}^\infty i {2i \choose i} (pq)^i = \frac{2pq}{(q-p)^3}$. See Appendix~\ref{appendix:sumclosedform} for a proof of this identity.
\end{proof}

We are now ready to bound $\E[|y_i-\tilde{y}_i|]$. We consider entries with value at most $\beta$, i.e., $y_i \leq m$.

\begin{lemma} \label{lem:ALP1yibound}
  Let $y_i \leq m$ and $\gamma=\frac{\alpha+2}{1+\frac{\alpha k}{s}}- 2$. Then the expected value of $|y_i - \tilde{y}_i|$ is bounded such that
  \begin{align*}
    \E[|y_i - \tilde{y}_i|] \leq \frac{4\alpha + 4}{\alpha^2}+\frac{4\gamma + 4}{\gamma^2} \enspace .
  \end{align*}
\end{lemma}

\begin{proof}
  Recall the definition of $P$ from Algorithm~\ref{alg:ALPestimator}. Let $\bar{y}_i \in P$ denote an element furthest from $y_i$ that is $|y_i-a| \leq |y_i - \bar{y}_i|$ for all $a \in P$. It it clearly sufficient to consider $\bar{y}_i$ for the proof since $|y_i - \tilde{y}_i| \leq |y_i - \bar{y}_i|$. 
  We first consider the case of $\bar{y}_i \leq y_i$. It follows from the definition of $\bar{y}_i$ as a maximum that $\sum_{j=\bar{y}_i+1}^{y_i} \tilde{z}_{h_j(i),j} \leq 0$. As such at least half the bits $(\tilde{z}_{h_{\bar{y}_i+1}(i),\bar{y}_i+1},\ldots,\tilde{z}_{h_{y_i}(i),y_i})$ must be zero, that is they were flipped by randomized response in Step~(3) of Algorithm~\ref{alg:ALP1}.
  As such the length of the longest interval ending at bit $\tilde{z}_{h_{y_i}(i),y_i}$ where at least half the bits were flipped is an upper bound on the value of $y_i-\bar{y}_i$.
  The expected size of said interval is bounded by the expected last non-negative step of a simple random walk with $p=\frac{1}{\alpha+2}$. It follows from Lemma~\ref{lem:randomwalk} that:
  \begin{align*}
    \E[y_i-\bar{y}_i \mid \bar{y}_i \leq y_i] & \leq \frac{4pq}{(q-p)^2} 
     = \frac{\frac{4\alpha + 4}{(\alpha+2)^2}}{\frac{\alpha^2}{(\alpha+2)^2}} 
     = \frac{4\alpha + 4}{\alpha^2} \enspace .
  \end{align*}
  We can use a similar argument when $y_i \geq \bar{y}_i$ to show that at least half the bits in $(\tilde{z}_{h_{y_i+1}(i),y_i+1},\ldots,\tilde{z}_{h_{\bar{y}_i}(i),\bar{y}_i})$ must be $1$ since $\bar{y}_i$ is a maximum. In this case we have to consider the possibility of hash collisions. Each hash function maps to $[s]$ and at most $k$ entries result in a hash collision. The probability of a hash collision is at most $\frac{k}{s}$ using a union bound. As such for $j > y_i$ we have $\Pr[\tilde{z}_{h_{j}(i),j} = 1] \leq (1- \frac{k}{s}) \cdot p + \frac{k}{s} \cdot q = \frac{1 + \frac{\alpha k}{s}}{\alpha + 2}$. We let $\frac{1 + \frac{\alpha k}{s}}{\alpha + 2}=\frac{1}{\gamma+2}$ such that $\E[\bar{y}_i - y_i \mid \bar{y}_i \geq y_i]\leq \frac{4\gamma + 4}{\gamma^2}$ by Lemma~\ref{lem:randomwalk} and the calculation above. We isolate $\gamma$ to find:
  \begin{align*}
    \frac{1}{\gamma + 2} &= \frac{1 + \frac{\alpha k}{s}}{\alpha + 2} \\
    \left(\Leftrightarrow\right) \quad \gamma + 2 &= \frac{\alpha + 2}{1 + \frac{\alpha k}{s}} \\ 
    \left(\Leftrightarrow\right)   \quad \phantom{+ 21} \gamma &= \frac{\alpha + 2}{1 + \frac{\alpha k}{s}} - 2 \enspace .
  \end{align*} 

  Note that $\gamma > 0$ holds due to the requirement $s > 2k$ of Algorithm~\ref{alg:ALP1}. By conditional expectation, we may upper bound the total expected error by
  \begin{align}\label{eq:error-bound}
    \E[|y_i-\tilde{y}_i|] & \leq \E[|y_i - \bar{y}_i|] \nonumber \\
    & \leq \E[y_i-\bar{y}_i \mid \bar{y}_i \leq y_i] + \E[\bar{y}_i-y_i \mid \bar{y}_i \geq y_i] \nonumber \\
    & \leq \frac{4\alpha + 4}{\alpha^2} + \frac{4\gamma + 4}{\gamma^2} \enspace .
  \end{align}
\end{proof}

As such we can bound the expected per-entry error for entries with a true value of at most $\beta$ by a function of the parameters $\alpha$ and $s$. In Section~\ref{sec:experiments} we discuss the choice of these parameters based on the upper bond and experiments. 
For any fixed values of $\alpha$ and $\frac{k}{s}$ we have:

\begin{lemma} \label{lem:ALP1fixederror}
  Let $\alpha=\Theta(1)$ and $s=\Theta(k)$. Then the expected per-entry error of Algorithm~\ref{alg:ALPestimator} is $\E[|x_i-\tilde{x}_i|] \leq \max(0, x_i-\beta) + O(1)$.
\end{lemma}

\begin{proof}
  It follows from Lemmas~\ref{lem:ALP1errorroughbound} and~\ref{lem:ALP1yibound} that the expected error for any entry bounded by $\beta$ is:
  \begin{align*}
    \E[\lvert x_i - \tilde{x}_i \rvert \mid x_i \leq \beta ] & \leq {\left(\frac{1}{2} + \frac{4\alpha + 4}{\alpha^2} + \frac{4\gamma + 4}{\gamma^2}\right)} \cdot \alpha \enspace ,
  \end{align*}
  
  where $\gamma=\frac{\alpha + 2}{1 + \frac{\alpha k}{s}}-2$. Entries above $\beta$ have an additional error of up to $x_i - \beta$, since $y_i = m$ and $y_i > m$ are represented identically in the embedding by Algorithm~\ref{alg:ALP1}. Since $\alpha$ and $\frac{k}{s}$ are constants we have:
  \begin{align*}
    \E[\lvert x_i - \tilde{x}_i \rvert] \leq \max(0, x_i - \beta) + O(1) \enspace .
  \end{align*}
\end{proof}

Next, we bound the tail probabilities for the per-entry error of the mechanism. We bound the error of the estimate $\tilde{y}_i$, which implies bounds on the error of the mechanism.

\begin{lemma} \label{lem:probboundonerror}
  Let $\gamma=\frac{\alpha+2}{1+\frac{\alpha k}{s}}-2$ and $\tau > 0$. 
  Let $p=\frac{1}{\gamma+2}$ and $q=1-p$. Then for Algorithm~\ref{alg:ALPestimator} we have:
  \begin{align*}
    \Pr[|y_i - \tilde{y}_i| \geq \tau] \leq \frac{2\cdot (4pq)^{\tau/2}}{\sqrt{\pi}(q-p)} \enspace ,
  \end{align*}
\end{lemma}

\begin{proof}
  Let $S$ be a simple random walk. We find an upper bound on the probability that the position of the last non-negative step in $S$ is at least $\tau$:
  \begin{align*}
    \Pr[(\max_n : S_n \geq 0) \geq \tau] & = \sum_{j=\ceil{\tau/2}}^\infty {2j \choose j} (pq)^j (q-p) \\
    & \leq \frac{q-p}{\sqrt{\pi}} \sum_{j=\ceil{\tau/2}}^\infty (4pq)^j \\
    & = \frac{q-p}{\sqrt{\pi}}\frac{(4pq)^{\ceil{\tau/2}}}{1-4pq} \\
    & \leq \frac{(4pq)^{\tau/2}}{\sqrt{\pi}(q-p)} \enspace ,
  \end{align*} 
  where the first inequality follows from ${2j \choose j} \leq \frac{4^j}{\sqrt{\pi j}}$ when $j \geq 1$~\cite{centralbinocoef}. The last inequality follows from $1-4pq=(q-p)^2$. As discussed in the proof of Lemma~\ref{lem:ALP1yibound}, the expectation of $|y_i - \tilde{y}_i|$ can be bounded by two random walks each with $p$ at most $\frac{1}{\gamma + 2}$. 
\end{proof}

\begin{lemma} \label{lem:maxerrorwithprob}
  Let $\gamma=\frac{\alpha+2}{1+\frac{\alpha k}{s}}-2$, $p=\frac{1}{\gamma+2}$ and $q=1-p$. With probability at least $1-\psi$ for Algorithm~\ref{alg:ALPestimator} we have:
  \begin{align*}
    |y_i - \tilde{y}_i| \leq \frac{2\log{\left(\frac{2}{\psi\sqrt{\pi}(q-p)}\right)}}{\log(1/(4pq))} \enspace .
  \end{align*}
\end{lemma}

\begin{proof}
  We set $\psi = \frac{2 \cdot (4pq)^{\tau/2}}{\sqrt{\pi}(q-p)}$ and isolate $\tau$ as follows:
  \begin{align*}
    \psi & = \frac{2 \cdot (4pq)^{\tau/2}}{\sqrt{\pi}(q-p)} \\
    (4pq)^{-\tau/2} & = \frac{2}{\psi\sqrt{\pi}(q-p)} \\
    \log(1/(4pq)) \cdot \frac{\tau}{2} & = \log{\left(\frac{2}{\psi\sqrt{\pi}(q-p)}\right)} \\
    \tau & = \frac{2\log{\left(\frac{2}{\psi\sqrt{\pi}(q-p)}\right)}}{\log(1/(4pq))} \enspace .
  \end{align*}
  By Lemma~\ref{lem:probboundonerror} we have: $\Pr[|y_i-\tilde{y}_i| \leq \tau] \geq 1-\psi$.
\end{proof}
Up to constant factors, the tail probabilities of our mechanism are similar to the properties of the Laplace mechanism summarized in Proposition~\ref{prop:laptail}. The probabilities depend on the parameters of the mechanism. In Section~\ref{sec:experiments} we fix the parameters and evaluate the error in practice. We summarize the tail probabilities for $|x_i-\tilde{x}_i|$ in Lemma~\ref{lem:ALP1tailbounds}.

\begin{lemma} \label{lem:ALP1tailbounds}
  Let $\gamma=\frac{\alpha+2}{1+\frac{\alpha k}{s}}-2$, $p=\frac{1}{\gamma+2}$, $q=1-p$, $x_i \leq \beta$, and $\tau \geq \alpha$. 
  Then for Algorithm~\ref{alg:ALPestimator} we have:
  \begin{align*}
    \Pr[|x_i - \tilde{x}_i| \geq \tau] < \frac{2 \cdot (4pq)^{(\tau/2\alpha)-1/2}}{\sqrt{\pi}(q-p)} \enspace ,
  \end{align*}
  With probability at least $1-\psi$ we have:
  \begin{align*}
    |x_i - \tilde{x}_i| < {\left(1 + \frac{2\log{\left(\frac{2}{\psi\sqrt{\pi}(q-p)}\right)}}{\log(1/(4pq))}\right)} \cdot \alpha \enspace .
  \end{align*}
\end{lemma}

\begin{proof}
  It is easy to see that $|x_i - x'_i| < (1+|y_i-\tilde{y}_i|) \cdot \alpha$ holds, as the error of random rounding is strictly less than $1$. The bounds follow from Lemmas~\ref{lem:probboundonerror} and~\ref{lem:maxerrorwithprob}.
\end{proof}

\subsection{Generalization to \texorpdfstring{$\varepsilon$}{𝜀}-differential privacy}

We now generalize the ALP mechanism from $1$-differential privacy to satisfying $\epsilon$-differential privacy. A natural approach is to use a function of $\epsilon$ as the parameter for randomized response in Algorithm~\ref{alg:ALP1}. The projection algorithm is $\epsilon$-differentially private if we remove the scaling step and set $p=\frac{1}{\epsilon+2}$. However, the expected per-entry error would be bounded by $\frac{8\epsilon+8}{\epsilon^2}$ by Equation~\ref{eq:error-bound} (without considering hash collisions), which is as large as $O{\left(\frac{1}{\epsilon^2}\right)}$ for small values of $\epsilon$. Other approaches modifying the value of $p$ have a similar expectation. 

In the following, we use a simple pre-processing and post-processing step to achieve optimal error. The idea is to scale the input vector as well as the parameter $\beta$ by $\epsilon$ before running Algorithm~\ref{alg:ALP1}. We scale back the estimates from Algorithm~\ref{alg:ALPestimator} by $1/\epsilon$. These generalizations are given as Algorithm~\ref{alg:ALPgeneral} and Algorithm~\ref{alg:ALPgeneralest}, respectively.

\begin{algorithm}[t]
  \caption{\ALPgeneral \label{alg:ALPgeneral}}
  \SetKwInOut{Parameters}{Parameters}\SetKwInOut{Input}{Input}\SetKwInOut{Output}{Output}
  \SetAlgoLined

  \Parameters{$\alpha, \beta, \epsilon > 0$, and $s \in \N$.}
  \Input{$k$-sparse vector $x \in \X$, where $s > 2k$. Sequence of hash functions from domain $[d]$ to $[s]$, $h = (h_1,\ldots,h_{m})$, where $m=\ceil{\frac{\beta\epsilon}{\alpha}}$.}
  \Output{$\epsilon$-differentially private representation of $x$.}

  \nl Scale the entries of $x$ such that $\hat{x}_i = x_i \cdot \epsilon$. \\
  \nl Let $h, \tilde{z} = \text{\ALPone}_{\alpha, \beta \cdot \epsilon, s}(\hat{x}, h)$. \\
  \nl Release $h$ and $\tilde{z}$.
\end{algorithm}

\begin{algorithm}[t]
  \caption{\ALPgeneralest \label{alg:ALPgeneralest}}
  \SetKwInOut{Parameters}{Parameters}\SetKwInOut{Input}{Input}\SetKwInOut{Output}{Output}
  \SetAlgoLined

  \Parameters{$\alpha, \epsilon > 0$.}
  \Input{Embedding $\tilde{z} \in \{0,1\}^{s \times m}$. Sequence of hash functions $h=(h_1,\ldots,h_{m})$. Index $i \in [d]$.}
  \Output{Estimate of $x_i$.}

  \nl Let $\tilde{x}_i = \text{\ALPest}_{\alpha}(\tilde{z}, h, i)$. \\
  \nl Return $\frac{\tilde{x}_i}{\epsilon}$. 
\end{algorithm}

\begin{lemma} \label{lem:ALPgeneralDP}
  Algorithm~\ref{alg:ALPgeneral} satisfies $\varepsilon$-differential privacy.
\end{lemma}

\begin{proof}
  It follows from the proof of Lemma~\ref{lem:ALP1DP} that for any subset of outputs $S$ we have $\frac{\Pr[\text{\ALPone}(\hat{x}') \in S]}{\Pr[\text{\ALPone}(\hat{x}) \in S]} \leq e^{\|\hat{x}-\hat{x}'\|_1}$. As such for any pair of neighboring vectors $x$ and $x'$ we have:
  \begin{align*}
    \frac{\Pr[\text{\ALPgeneral}(x') \in S]}{\Pr[\text{\ALPgeneral}(x) \in S]} & = \frac{\Pr[\text{\ALPone}(\hat{x}') \in S]}{\Pr[\text{\ALPone}(\hat{x}) \in S]} \\
    & \leq e^{\|\hat{x} - \hat{x}'\|_1} %
     = e^{\epsilon \cdot \|x - x'\|_1} %
     \leq e^\epsilon \enspace .
  \end{align*}
\end{proof}

\begin{lemma} \label{lem:ALPgeneralsize}
  Let $\alpha=\Theta(1)$ and $s=\Theta(k)$. The output of Algorithm~\ref{alg:ALPgeneral} can be stored using $O(k\beta\epsilon)$ bits.
\end{lemma}

\begin{proof}
  It follows from Lemma~\ref{lem:ALP1space} that the output can be stored using $O\left(\frac{(s+\log(d))\beta\epsilon}{\alpha}\right)$ bits. Recall that we assume $k=\Omega(\log(d))$, i.e., $O\left(\frac{(s+\log(d))\beta\epsilon}{\alpha}\right)=O(k\beta\epsilon)$.
\end{proof}

\begin{lemma} \label{lem:ALPgeneralerror}
  Let $\alpha=\Theta(1)$ and $s=\Theta(k)$. Then the expected per-entry error of Algorithm~\ref{alg:ALPgeneralest} is $\E[\lvert x_i - \tilde{x}_i \rvert] \leq \max(0, x_i - \beta) + O(1/\varepsilon)$ and the evaluation time is $O(\beta\epsilon)$.
\end{lemma}

\begin{proof}
  It is clear that the error of Algorithm~\ref{alg:ALPgeneralest} is $\frac{1}{\epsilon}$ times the error of Algorithm~\ref{alg:ALPestimator} for entries at most $\beta$. As such the expected per-entry error follows from Lemma~\ref{lem:ALP1fixederror}.
The evaluation time follows directly from Lemma~\ref{lem:ALP1estrunningtime}.
\end{proof}

\begin{lemma} \label{lem:ALPgeneraltailbounds}
  Let $\gamma=\frac{\alpha+2}{1+\frac{\alpha k}{s}}-2$, $p=\frac{1}{\gamma+2}$, $q=1-p$, $x_i \leq \beta$, and $\tau \geq \frac{\alpha}{\epsilon}$. 
  Then for Algorithm~\ref{alg:ALPgeneralest} we have:
  \begin{align*}
    \Pr[|x_i - \tilde{x}_i| \geq \tau] < \frac{2 \cdot (4pq)^{(\tau\epsilon/2\alpha)-1/2}}{\sqrt{\pi}(q-p)} \enspace ,
  \end{align*}
  With probability at least $1-\psi$ we have:
  \begin{align*}
    |x_i - \tilde{x}_i| < {\left(1 + \frac{2\log{\left(\frac{2}{\psi\sqrt{\pi}(q-p)}\right)}}{\log(1/(4pq))}\right)} \cdot \frac{\alpha}{\epsilon} \enspace .
  \end{align*}
\end{lemma}

\begin{proof}
  It follows directly from Lemma~\ref{lem:ALP1tailbounds}.
\end{proof}

We are now ready to state the following theorem which summarizes the properties of the ALP mechanism. 

\begin{theorem} \label{thm:ALPplain}
  Let $\alpha=\Theta(1)$, $s=\Theta(k)$. Then there exist an algorithm where the expected per-entry error is $O(1/\epsilon)$ for all entries, the access time is $O(u\epsilon)$, and the space usage is $O(k u \epsilon)$ bits.
\end{theorem}

\begin{proof}
  By setting $\beta=u$, it follows directly from Lemmas~\ref{lem:ALPgeneralsize} and~\ref{lem:ALPgeneralerror}.
\end{proof}

The space usage and access time of the mechanism both scale linearly with the parameter $\beta$. As such the mechanism performs well only for small values of $u$. However, in many contexts $u$ scales with the input size. One example is a histogram, where $u$ is the number of rows in the underlying dataset. Next, we show how to handle such cases.

\section{Combined data structure}
\label{sec:combineddatastructure}

In this section, we combine the ALP mechanism with techniques from previous work to improve space requirements and access time. 
As shown in Theorem~\ref{thm:ALPplain} the ALP mechanism performs well when all entries are bounded by a small value. The per-entry error is low only for entries bounded by $\beta$ but the space requirements and access time scale linearly with $\beta$. 
Some of the algorithms from previous work perform well for large entries but have large per-entry error for small values.
The idea of this section is to combine the ALP mechanism with such an algorithm to construct a composite data structure that performs well for both small and large entries.

To handle large values, we use the thresholding technique from Cormode et al.~\cite{cormode2012differentially}. It adds noise to each entry, but only stores entries above a threshold. The pseudocode of the algorithm is given as Algorithm~\ref{alg:threspure}. 

\begin{algorithm}[t]%
  \caption{\threshold~\cite{cormode2012differentially} \label{alg:threspure}}
  \SetKwInOut{Parameters}{Parameters}\SetKwInOut{Input}{Input}\SetKwInOut{Output}{Output}
  \SetAlgoLined
  
  \Parameters{$\epsilon, t > 0$.}
  \Input{$k$-sparse vector $x \in \X$.}
  \Output{$\epsilon$-differentially private representation of $x$.}
  
  \nl Let $v_i = x_i + \eta_i$ for all $i \in [d]$, where $\eta_i \sim \Lap{1/\epsilon}$. \\
  \nl Truncate entries below $t$:
  $$\tilde{v}_i = 
  \begin{cases}
    v_i, & \text{if } y_i \geq t \\
    0, & \text{otherwise} \\
  \end{cases}$$
  \nl Return $\tilde{v}$.
\end{algorithm}

\begin{lemma} \label{lem:thresholdpriv}
  Algorithm~\ref{alg:threspure} satisfies $\epsilon$-differential privacy. 
\end{lemma}

\begin{proof}
  The algorithm is equivalent to the Laplace mechanism followed by post-processing. The Laplace mechanism satisfies $\epsilon$-differential privacy, and privacy is preserved under post-processing as stated by Lemma~\ref{lem:postprocessing}.
\end{proof}

\begin{lemma} \label{lem:thresholdsize}
  Let $t=\frac{\ln(d/2)}{\epsilon}$. Then the output of Algorithm~\ref{alg:threspure} is $O(k)$-sparse with high probability. 
\end{lemma}

\begin{proof}
  Using Definition~\ref{def:lapcdf} we find that the probability of storing a zero entry is:
  \begin{align*}
    \Pr[\Lap{1/\epsilon} \geq t] & = \Pr[\Lap{1/\epsilon} \leq -t] %
     = \frac{1}{2}e^{-t\epsilon} %
     = \frac{1}{d} \enspace .
  \end{align*}
  By linearity of expectation, the expected number of stored true zero entries is at most one, and as such the vector is $O(k)$-sparse with high probability.
\end{proof}

As discussed in Section~\ref{sec:relatedwork}, the expected per-entry error of Algorithm~\ref{alg:threspure} is $\Ologdeps$ for worst-case input. We combine the algorithm with the ALP mechanism from the previous section to achieve $\Ooneeps$ expected per-entry error for any input. We use the threshold parameter $t$ as value for parameter $\beta$ in Algorithm~\ref{alg:ALPgeneral}. The algorithm is presented in Algorithm~\ref{alg:ALPcombined}. 

\begin{algorithm}[t]
  \caption{\ALPfinal \label{alg:ALPcombined}}
  \SetKwInOut{Parameters}{Parameters}\SetKwInOut{Input}{Input}\SetKwInOut{Output}{Output}
  \SetAlgoLined

  \Parameters{$\alpha, \epsilon_1, \epsilon_2 > 0$, and $s \in \N$.}
  \Input{$k$-sparse vector $x \in \X$, where $s > 2k$. 
  Sequence of hash functions from domain $[d]$ to $[s]$, $h = (h_1,\ldots,h_{m})$, where $m=\ceil{\frac{\beta\epsilon_2}{\alpha}}$.} 
  \Output{$(\epsilon_1 + \epsilon_2)$-differentially private representation of $x$.}

  \nl Let $t=\frac{\ln(d/2)}{\epsilon_1}$. \\
  \nl Let $\tilde{v} = \text{\threshold}_{\epsilon_1, t}(x)$. \\
  \nl Let $h, \tilde{z} = \text{\ALPgeneral}_{\alpha, \epsilon_2, t, s}(x, h)$ \\
  \nl Return $\tilde{v}$, $h$ and $\tilde{z}$.
\end{algorithm} 

\begin{lemma} \label{lem:ALPcombinedDP}
  Algorithm~\ref{alg:ALPcombined} satisfies $(\epsilon_1 + \epsilon_2)$-differential privacy. 
\end{lemma}

\begin{proof}
  The two parts of the algorithm are independent as there is no shared randomness. The first part of the algorithm satisfies $\epsilon_1$-differential privacy by Lemma~\ref{lem:thresholdpriv} and the second part satisfies $\epsilon_2$-differential privacy by Lemma~\ref{lem:ALPgeneralDP}. As such it follows directly from composition (Lemma~\ref{lem:composition}) that Algorithm~\ref{alg:ALPcombined} satisfies $(\epsilon_1 + \epsilon_2)$-differential privacy.  
\end{proof}

\begin{lemma} \label{lem:ALPcombinedsize}
  Let $\alpha=\Theta(1)$, $s=\Theta(k)$, $\epsilon_1=\Theta(\epsilon_2)$. Then the output of Algorithm~\ref{alg:ALPcombined} is stored using $O(k\log(d+u))$ bits with high probability.  
\end{lemma}

\begin{proof}
  It follows from Lemma~\ref{lem:thresholdsize} that we can store $\tilde{v}$ using $O(k\log(d + u))$ bits with high probability. Since $\beta=t$ it follows from Lemma~\ref{lem:ALPgeneralsize} that we can store $h$ and $\tilde{z}$ using $O(kt\epsilon_2)=O(k\log(d))$ bits. 
\end{proof}

To estimate an entry, 
we access $\tilde{v}$ when a value is stored for the entry and the ALP embedding otherwise. This algorithm is presented in Algorithm~\ref{algo:ALPcombinedest}.

\begin{algorithm}[t]
  \caption{\ALPfinalest \label{algo:ALPcombinedest}}
  \SetKwInOut{Parameters}{Parameters}\SetKwInOut{Input}{Input}\SetKwInOut{Output}{Output}
  \SetAlgoLined

  \Parameters{$\alpha, \epsilon_2 > 0$.}
  \Input{Vector $\tilde{v} \in \X$. Embedding $\tilde{z} \in \{0,1\}^{s \times m}$. Sequence of hash functions $h=(h_1,\ldots,h_{m})$. Index $i \in [d]$.}
  \Output{Estimate of $x_i$.}

  \nl Estimate the entry using either the vector or the embedding such that:
  $$\tilde{x}_i = \begin{cases}
    \tilde{v}_i, & \text{if } \tilde{v}_i \neq 0 \\
    \text{\ALPgeneralest}_{\epsilon_2, \alpha}(\tilde{z}, h, i), & \text{otherwise}
  \end{cases}$$
  \nl Return $\tilde{x}_i$.
\end{algorithm}

\begin{lemma} \label{lem:ALPpure}
  Let $\alpha=\Theta(1)$, $s=\Theta(k)$, and $\epsilon_1=\Theta(\epsilon_2)$.
  Let $\tilde{v}$, $h$, and $\tilde{z}$ be the output of Algorithm~\ref{alg:ALPcombined} given these parameters. Then the evaluation time of Algorithm~\ref{algo:ALPcombinedest} is $O(\log(d))$. The expected per-entry error is $O(1/\epsilon)$ and the expected maximum error is $\Ologdeps$.
\end{lemma}

\begin{proof}
  The evaluation time follows from Lemma~\ref{lem:ALPgeneralerror}. That is, the evaluation time is $O(\beta\epsilon)=O(t\epsilon)=O(\log(d))$.

  The error depends on both parts of the algorithm. The expected per-entry error for the $i$th entry is $\max(0,x_i-\beta) + O(1/\epsilon_2)$ when $\tilde{v}_i = 0$ by Lemma~\ref{lem:ALPgeneralerror}. That is, when $\eta_i$ is less than $\beta-x_i$ in Algorithm~\ref{alg:threspure}. When $\tilde{v}_i \neq 0$ the error is the absolute value of $\eta_i$. That is, we can analyze it using conditional probability and the probability density function of the Laplace distribution from Definition~\ref{def:lappdf}.
  \begin{align*}
    \E[|x_i-\tilde{x}_i|] & = \E[|x_i-\tilde{x}_i| \mid \tilde{v}_i = 0] \cdot \Pr[\tilde{v}_i=0] \\
    & \quad + \E[|x_i-\tilde{x}_i| \mid \tilde{v}_i \neq 0] \cdot \Pr[\tilde{v}_i \neq 0] \\
    & \leq (\max(0,x_i-\beta) + O(1/\epsilon_2)) \cdot \Pr[\Lap{1/\epsilon_1} < \beta-x_i]  \\
    & \quad + \int_{\beta - x_i}^\infty |v - x_i| \cdot \frac{\epsilon_1}{2} e^{-\lvert v - x_i \rvert \epsilon_1} \,dv \\
    & < \int_{-\infty}^{\beta - x_i} (|v-x_i| + O(1/\epsilon_2)) \cdot \frac{\epsilon_1}{2} e^{-|v - x_i|\epsilon_1} \,dv \\
    & \quad + \int_{\beta - x_i}^\infty |v - x_i| \cdot \frac{\epsilon_1}{2} e^{-\lvert v - x_i \rvert \epsilon_1} \,dv \\
    & < \int_{-\infty}^\infty |v - x_i| \cdot \frac{\epsilon_1}{2} e^{-|v - x_i|\epsilon_1} \,dv + O(1/\epsilon_2) \\
    & = O(1/\epsilon_1) + O(1/\epsilon_2) = O(1/\epsilon) \enspace .
  \end{align*}
  The expected maximum error of Algorithm~\ref{alg:threspure} is $\Ologdeps$ and the output of the Algorithm~\ref{alg:ALPgeneralest} is at most $\beta$. Since $\beta=\Ologdeps$ the expected maximum error is $\Ologdeps$. 
\end{proof}

\subsection{Removing the dependency on dimension}

To make access time independent of the dimension $d$, we can turn to approximate differential privacy. 
This allows us to use a smaller threshold in the initial thresholding approach, which in turn results in smaller values for $\beta$ in the ALP mechanism.

The following algorithm is similar to that introduced by Korolova et al.~\cite{korolova2009releasing}, which we discussed in Section~\ref{sec:relatedwork}. It adds noise to non-zero entries only, and uses a threshold to satisfy approximate differential privacy. Our algorithm differs from the work of Korolova et al. by using a random rounding step. This step is not needed in a discrete setting, where at most a single zero-valued entry is changed to a non-zero entry for neighboring vectors. However, in the real-valued context, several zero entries can change.

\begin{algorithm}[t]
  \caption{\thresholdapprox~(Following technique by~\cite{korolova2009releasing}) \label{alg:thresapprox}}
  \SetKwInOut{Parameters}{Parameters}\SetKwInOut{Input}{Input}\SetKwInOut{Output}{Output}
  \SetAlgoLined
  
  \Parameters{$\epsilon, \delta > 0$.}
  \Input{$k$-sparse vector $x \in \X$.}
  \Output{$(\epsilon, \delta)$-differentially private approximation of $x$.}
  
  \nl Apply random rounding to non-zero entries below $1$ such that:
  $$y_i = \begin{cases}
    \RRound{x_i}, & \text{if } 0 < x_i < 1 \\
    x_i, & \text{otherwise}
  \end{cases}$$ \\
  \nl Let $v_i = y_i + \eta_i$ for all non-zero entries, where $\eta_i \sim \Lap{1/\epsilon}$. \\
  \nl Let $t=\frac{\ln{\left(1/\delta\right)}}{\epsilon}+2$. \\
  \nl Truncate entries below $t$:
  $$\tilde{v}_i = 
  \begin{cases}
    v_i, & \text{if } y_i \neq 0 \text{ and } v_i \geq t \\
    0, & \text{otherwise} 
  \end{cases}$$ \\
  \nl Return $\tilde{v}$.
\end{algorithm}

\begin{lemma} \label{lem:alpcombinedpureDP}
  Algorithm~\ref{alg:thresapprox} satisfies $(\epsilon, \delta)$-differential privacy. 
\end{lemma}

\begin{proof}
  Let $x$ and $x'$ be neighboring vectors. We consider two additional vectors $\hat{x}$ and $\hat{x}'$ such that:
  $$\hat{x}_i = \begin{cases}
    \min(1, x'_i), & \text{if } x_i \leq 1 \\
    x_i, & \text{otherwise;}
  \end{cases}$$
  $$\hat{x}'_i = \begin{cases}
    1, & \text{if } x'_i < 1 \text{ and } 1 < x_i \\
    x'_i, & \text{otherwise.}
  \end{cases}$$
  The vectors are constructed such that $x$ and $\hat{x}$ can only differ for entries at most $1$ in both vectors. The same holds for $x'$ and $\hat{x}'$. Additionally, the $\ell_1$-distance is still at most $1$ between any pair of vectors. 
  
  We find the probability of outputting anything for an entry less than or equal to $1$ as:
  \begin{align*}
    \Pr[\tilde{v}_i \neq 0 | x_i \leq 1] & = \Pr[y_i=1] \cdot \Pr[\Lap{1/\epsilon} \geq t-1] \\
    & = x_i \cdot \Pr[\Lap{1/\epsilon} \leq -(t - 1)] \\
    & = x_i \cdot \frac{1}{2} e^{-(t-1)\epsilon} \\
    & = x_i \cdot \frac{1}{2} e^{-\ln(1/\delta)-\epsilon} \\
    & = x_i \frac{\delta}{2\cdot e^\epsilon} \enspace .
  \end{align*}
  
  Since $x$ and $\hat{x}$ only differ for entries less than or equal to $1$ we have for any subset of outputs $S$:
  \begin{align*}
    \Pr[\text{\thresholdapprox}(x) \in S] & \leq \Pr[\text{\thresholdapprox}(\hat{x}) \in S] \\
    & \quad + \sum_{i \in [d]} |\hat{x}_i-x_i| \frac{\delta}{2\cdot e^\epsilon} \\
    & \leq \Pr[\text{\thresholdapprox}(\hat{x}) \in S] + \frac{\delta}{2 \cdot e^\epsilon} \enspace .
  \end{align*}
  The inequality holds in both directions and for the pair of $x'$ and $\hat{x}'$ as well. 
  
  By definition $\hat{x}$ and $\hat{x}'$ only differ for entries of at least $1$. As such we can ignore the random rounding step and we have:
  \begin{align*}
    \Pr[\text{\thresholdapprox}(\hat{x}) \in S] &\leq e^{\|\hat{x}-\hat{x}'\|_1\epsilon} \Pr[\text{\thresholdapprox}(\hat{x}') \in S] \\
    & \leq e^\epsilon \cdot \Pr[\text{\thresholdapprox}(x) \in S] \enspace .
  \end{align*}

  Using the inequalities above we have:
  \begin{align*}
    \Pr[\text{\thresholdapprox}(x) \in S] & \leq \Pr[\text{\thresholdapprox}(\hat{x}) \in S] + \frac{\delta}{2e^\epsilon} \\
    & \leq e^\epsilon \cdot \Pr[\text{\thresholdapprox}(\hat{x}') \in S] + \frac{\delta}{2 \cdot e^\epsilon} \\
    & \leq e^\epsilon \cdot \left(\Pr[\text{\thresholdapprox}(x') \in S] + \frac{\delta}{2 \cdot e^\epsilon}\right) \\
    & \quad + \frac{\delta}{2 \cdot e^\epsilon} \\
    & \leq e^\epsilon \cdot \Pr[\text{\thresholdapprox}(x') \in S] + \delta \enspace .
  \end{align*}
\end{proof}

\begin{lemma}
  Let $\delta=O\left(\frac{1}{k}\right)$. Then the expected maximum error of Algorithm~\ref{alg:thresapprox} is $\Ologdeltaeps$.
\end{lemma}

\begin{proof}
  The expected maximum error added by the Laplace noise is $O\left(\frac{\log(k)}{\epsilon}\right)$, since we add noise to at most $k$ entries. By removing entries we add error of up to $\Ologdeltaeps$. As such the expected maximum error for worst-case input is:
  \begin{align*}
    \E[\|x-\tilde{v}\|_\infty] & \leq O\left(\frac{\log(k)}{\epsilon}\right) + O\left(\frac{\log(1/\delta)}{\epsilon}\right) 
    = O\left(\frac{\log(1/\delta)}{\epsilon}\right) \enspace .
  \end{align*}
\end{proof}

In the following, we use Algorithm~\ref{alg:thresapprox} instead of Algorithm~\ref{alg:threspure} in Step~(2) of Algorithm~\ref{alg:ALPcombined}.
\begin{lemma} \label{lem:ALPapprox}
  Let $\alpha=\Theta(1)$, $s=\Theta(k)$, $\epsilon_1=\Theta(\epsilon_2)$, and $\delta > 0$ with $\delta = O(1/k)$.
  By using Algorithm~\ref{alg:thresapprox} in Algorithm~\ref{alg:ALPcombined} the access time is $O(\log(1/\delta))$. The expected per-entry error is $O(1/\epsilon)$ and the expected maximum error is $\Ologdeltaeps$. The combined mechanism satisfies $(\epsilon_1 + \epsilon_2,\delta)$-differential privacy.
\end{lemma}

\begin{proof}
  The proof is the same as the proofs of Lemmas~\ref{lem:ALPcombinedDP} and~\ref{lem:ALPpure}.
\end{proof}

\begin{lemma} \label{lem:ALPapproxsize}
  Let $\alpha=\Theta(1)$, $s=\Theta(k)$, and $\epsilon_1=\Theta(\epsilon_2)$. Then the memory requirement of combining Algorithm~\ref{alg:thresapprox} and the ALP mechanism is $O(k(\log(d+u) + \log(1/\delta)))$.
\end{lemma}

\begin{proof}
  The output of Algorithm~\ref{alg:thresapprox} is always $k$-sparse and can be represented using $O(k\log(d+u))$ bits. We set $\beta=\frac{\ln(1/\delta)}{\epsilon_2}+2$ and therefore $h$ and $\tilde{z}$ are represented using $O(k\log(1/\delta))$ bits by Lemma~\ref{lem:ALPgeneralsize}.
\end{proof}

We are now ready to summarize our results for both pure and approximate differential privacy.

\begin{theorem} \label{thm:ALPpure}
  Let $\alpha=\Theta(1)$, $s=\Theta(k)$, and $\epsilon > 0$. Then there exists an $\epsilon$-differentially private algorithm with $\Ooneeps$ expected per-entry error, $\Ologdeps$ expected maximum error, access time of $O(\log(d))$, and space usage of $O(k\log(d+u))$ with high probability.
\end{theorem}

\begin{proof}
  It follows directly from Lemmas~\ref{lem:ALPcombinedDP},~\ref{lem:ALPcombinedsize} and~\ref{lem:ALPpure}.
\end{proof}

\begin{theorem} \label{thm:ALPapprox}
  Let $\alpha=\Theta(1)$, $s=\Theta(k)$, and $\epsilon, \delta > 0$. Then there exist an $(\epsilon, \delta)$-differentially private algorithm with $\Ooneeps$ expected per-entry error, $\Ologdeltaeps$ expected maximum error, access time of $O(\log(1/\delta))$, and space usage of $O(k(\log(d+u)+\log(1/\delta)))$.
\end{theorem}

\begin{proof}
  It follows directly from Lemmas~\ref{lem:alpcombinedpureDP},~\ref{lem:ALPapprox} and~\ref{lem:ALPapproxsize}.
\end{proof}

\section{Experiments}
\label{sec:experiments}

\begin{figure*}[t]
  \subfigure[Upper bound on expected per-entry error]{\label{fig:plotupperbound}\includegraphics[width=0.33\linewidth]{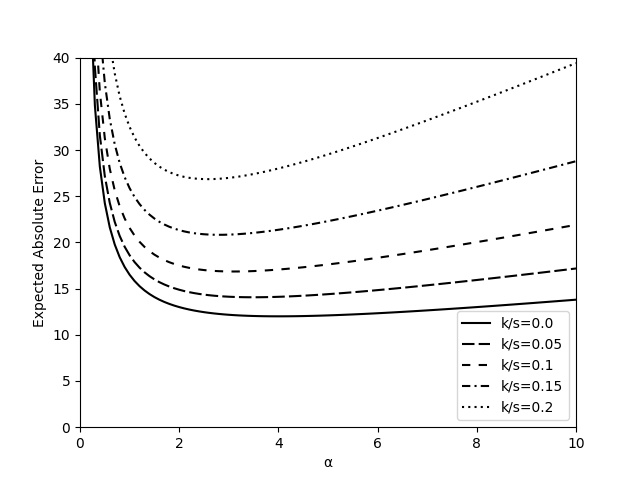}} 
  \subfigure[Observed mean per-entry error]{\label{fig:plotexperiments}\includegraphics[width=0.33\linewidth]{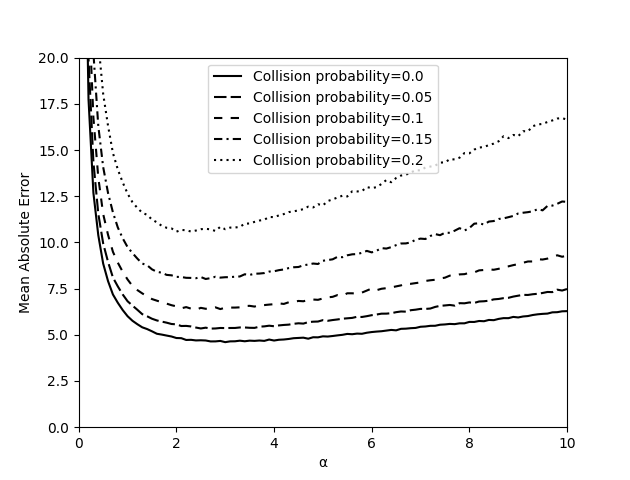}}
  \subfigure[Error distribution $\alpha=3$ and $\mathrm{collisions}=0.1$.]{\label{fig:errordist}\includegraphics[width=0.33\linewidth]{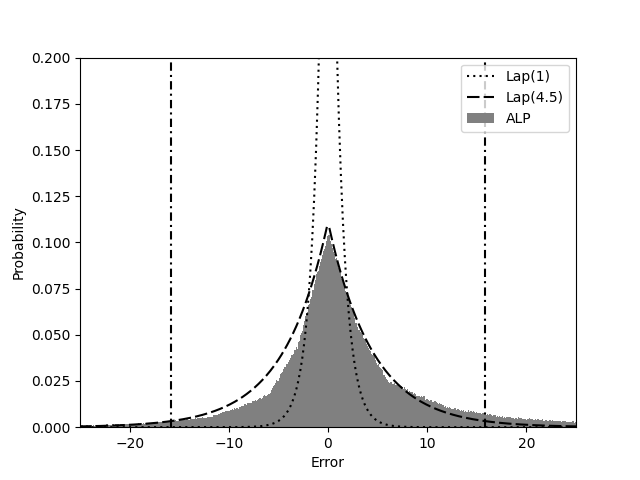}}
  \captionsetup{justification=centering}
  \caption{ Theoretical expected per-entry error and experiment results. \\
            Note that the y-axes for the plots use different scales.}
  \label{fig:perentrycomparison}
  \Description[Plot of the upper bound and observed per-entry error]{The plots show the expected per-entry error for varying parameters of $\alpha$ and $k/s$. The plots follow similar trajectories, but the observed error is approximately half the theoretical upper bound.}
\end{figure*}

In this section, we discuss the per-entry error of \ALPest{} (Algorithm~\ref{alg:ALPestimator}) in practice.
Let $\gamma=\frac{\alpha+2}{1+\frac{\alpha k}{s}}-2$. By Lemma~\ref{lem:ALP1errorroughbound} and~\ref{lem:ALP1yibound} the expected per-entry error of \ALPest{} is upper bounded by:
\begin{align*}
\E[|x_i-\tilde{x}_i|] \leq \left(\frac{1}{2} + \frac{4\alpha + 4}{\alpha^2} + \frac{4\gamma + 4}{\gamma^2}\right) \cdot \alpha \enspace .
\end{align*}

Figure~\ref{fig:plotupperbound} shows the upper bound for varying values of $k/s$ and $\alpha$. Recall that $k/s$ is a bound on the probability of a hash collision. We see that the effect of hash collisions on the error increases for large values of $\alpha$, as each bit in the embedding is more significant. We discuss how the upper bound compares to practice next.

\paragraph{Experimental Setup.}
We designed experiments to evaluate the effect of the adjustable parameters $\alpha$ and $s$ on the expected per-entry error of \ALPest. The experiments were performed on artificial data.
For our setup, we set parameter $\beta=5000$ and chose a true value $x_i$ uniformly at random in the interval $[0,\ldots,\beta]$. We run only on artificial data, as uniform data does not benefit the algorithm, and we can easily simulate worst-case conditions for hash collisions. 
We simulate running the \ALPone{} algorithm by computing $y_i$, simulating hash collisions, and applying randomized response. The probability for hash collisions is fixed in each experiment and the same probability is used for all bits. This simulates worst-case input in which all other non-zero entries have a true value of at least $\beta$. We increment $\alpha$ by steps of $0.1$ in the interval $[0.1,\ldots,10]$ and the probability of a hash collision by $0.05$ in the interval $[0,\ldots,0.2]$. The probability of $0$ serves only as a baseline, as it is not achievable in practice for $k > 1$. The experiment was repeated $10^5$ times for every data point.

Figure~\ref{fig:plotexperiments} shows plots of the mean absolute error of the experiments. As $\alpha$ is increased, the error drops off at first and slowly climbs. The estimates of $y_i$ are more accurate for large values of $\alpha$. However, any inaccuracy is more significant, as $\tilde{y}_i$ is scaled back by a larger value. The error from the random rounding step also increases with $\alpha$.
The plots of the upper bound and observed error follow similar trajectories. However, the upper bound is approximately twice as large for most parameters.

\paragraph{Fixed Parameters.}
The experiments show how different values of $\alpha$ and $s$ affect the expected per-entry error. However, the parameters also determine constant factors for space usage and access time. The space requirements scale linearly in $\frac{s}{\alpha}$ and the access time is inversely proportional to $\alpha$. As such, the optimal parameter choice depends on the use case due to space, access time, and error trade-offs. 

To evaluate the error distribution of the \ALPest{} algorithm we fixed the parameters of an experiment. We set $\alpha=3$ and the hash collision probability to $0.1$. We repeated the experiment $10^6$ times.

The error distribution is shown in Figure~\ref{fig:errordist}. The mean absolute error of the experiment is $6.4$ and the standard deviation is $11$. 
Plugging in the parameters in Lemma~\ref{lem:ALP1tailbounds},
with probability at least 0.9 the error is at most
\begin{align*}
  |x_i-\tilde{x}_i| < 3 + \frac{6\log{\left(\frac{5}{0.12\sqrt{\pi}}\right)}}{\log{\left(\frac{25}{19.24}\right)}} \approx 75.33 \enspace .
\end{align*}
The error of the observed $90$th percentile is $15.78$, which is shown in Figure~\ref{fig:errordist} using vertical lines. 
Again, this shows that the upper bounds are pessimistic.

For comparison, the plots include the Laplace distribution with scale parameters $1$ and $4.5$. Note that the Laplace distribution with parameter $1$ is optimal for the privacy budget. The standard deviation of the distribution with scale $4.5$ is $6.36$ and as such the mean absolute error is similar to the ALP mechanism. 

The distribution is slightly off-center, and the mean error is $2.33$. This is expected due to hash collisions. The effect of hash collisions is also apparent for the largest observed errors. The lowest observed error was $-114$, while the highest was $274$. 
There is a clear trade-off between space usage and per-entry error. We reran the experiment with hash collision probability $0.01$ using the same value for $\alpha$. The error improved for all the metrics mentioned above. The mean absolute error is $4.8$, the standard deviation is $7.8$, the mean error is $0.18$, the $90$th percentile is $11.5$, and the largest observed error is $147$.

\section{Suggestions to Practitioners}
\label{sec:practioners}
The ALP mechanism introduced in this paper combines the best of three worlds: It has low error similar to the Laplace mechanism, produces compact representations using asymptotically optimal space, and has an access time that scales only with $O(\log d)$.

In an application that wants to make use of differentially private histograms/vectors, one first has to get an overview of the assumed properties of the data before making a choice on which approach to use. 
If $d$ is small or the data is assumed to be dense, the Laplace mechanism will offer the best performance. 
If the data is sparse and the dimension $d$ is large, the analyst must know which error guarantee she wishes to achieve, and which access time is feasible in the setting where the application is deployed. 
If a larger error is acceptable for ``small'' entries or access time is crucial, just applying the thresholding technique~\cite{korolova2009releasing,cormode2012differentially} is the better choice.
Otherwise, if small error is paramount or an access time of $O(\log d)$ is sufficient, the ALP mechanism will provide the best solution. 

\paragraph{Variants.}
We assume in this paper that $k$ is a known bound on the sparsity of the input data. However, in some applications the value of $k$ itself is private. Here we briefly discuss approaches in such settings. We use the value of $k$ to select the size of the embedding, such that the probability of hash collisions is sufficiently small. When $k$ is not known we can still bound the probability of hash collisions.

If the input is a \emph{histogram} the sparsity differs by at most $1$ for neighboring datasets. As such we can use a fraction of the privacy budget to estimate the sparsity. Note that this is not possible for vectors, as the difference in sparsity can be as large as $d$ for neighboring datasets.

If $\|x\|_1=n$ is known then we have $\|\hat{x}\|_1=n\epsilon$ for the scaled input. We can bound the probability of hash collisions by a constant when the size of the embedding is $O(n\epsilon)$ bits. If $\|x\|_1$ is unknown we can estimate it using a fraction of the privacy budget. Note that the space differs from the $k$-sparse setting, and remains $O(n\epsilon)$ bits when applying the thresholding techniques.

An implementation of a variant of the ALP mechanism is available as part of the open source project OpenDP (\url{https://opendp.org/}) in the repository \url{https://github.com/opendp/opendp}.

\section{Open problems}\label{sec:openproblems}

The main open problem that we leave is if it is possible to achieve similar space and error with constant time access.
We know of a way (based on the count-min sketch) to achieve optimal \emph{expected} error with constant time access and space within a logarithmic factor of optimal. However, this method does not have strong tail bounds on the error.

\begin{acks}
  We thank the anonymous reviewers for their detailed suggestions that helped improve the paper.
  Christian Janos Lebeda and Rasmus Pagh are affiliated with Basic Algorithms Research Copenhagen (BARC), supported by the VILLUM Foundation grant 16582.
\end{acks}

\bibliographystyle{ACM-Reference-Format}
\bibliography{biblio}

\newpage
\appendix
\section{Closed-form proof of Lemma~\ref{lem:randomwalk}}
\label{appendix:sumclosedform}
Here we provide a closed-form expression used in the proof of Lemma~\ref{lem:randomwalk}. 

In the proof, we will make use of general binomial coefficient(\cite[Equation~5.1]{concretemath}):
\begin{align*}
  {r \choose k} = \frac{r(r-1)\ldots(r-k+2)(r-k+1)}{k!} \enspace ,
\end{align*} 
and the binomial theorem~(\cite[Equation~5.12]{concretemath}): 
\begin{align*} 
  (1+z)^r=\sum_{k=0}^\infty {r \choose k}(z)^k \enspace .
\end{align*} 

Starting from an infinite series with $z < 1/4$, we simplify as follows:
\begin{align*}
  \sum_{k=0}^\infty k {2k \choose k} (z)^k & = \sum_{k=1}^\infty k \frac{(2k)!}{k!k!} z^k \\
  & = \sum_{k=1}^\infty k \frac{k(k-\frac{1}{2})(k-1)\ldots{\left(\frac{3}{2}\right)} 1  {\left(\frac{1}{2}\right)}}{k!k!} 2^{2k} z^k \\
  & = \sum_{k=1}^\infty \frac{(k-\frac{1}{2}) (k-\frac{3}{2})\ldots{\left(\frac{5}{2}\right)}{\left(\frac{3}{2}\right)}{\left(\frac{1}{2}\right)}}{(k-1)!} (4z)^k \\
  & = \frac{4z}{2} \sum_{k=1}^\infty \frac{(k-\frac{1}{2}) (k-\frac{3}{2})\ldots{\left(\frac{5}{2}\right)}{\left(\frac{3}{2}\right)}}{(k-1)!} (4z)^{k-1} \\
  & = 2z \sum_{k=1}^\infty \frac{{\left(-\frac{3}{2}\right)}{\left(-\frac{5}{2}\right)}\ldots(-k+\frac{3}{2}) (-k+\frac{1}{2})}{(k-1)!} (-4z)^{k-1} \\
  & = 2z \sum_{k=1}^\infty {- \frac{3}{2} \choose k - 1} (-4z)^{k-1} \\
  & = 2z \sum_{k=0}^\infty {- \frac{3}{2} \choose k} (-4z)^{k} \\
  & = \frac{2z}{(1-4z)^{3/2}} \enspace .
\end{align*}

Let $p=\frac{a}{a+b}$ and $q=\frac{b}{a+b}$. Then we have: 
\begin{align*}
  1-4pq & = \frac{(a+b)^2}{(a+b)^2} - \frac{4ab}{(a+b)^2} \\
  & = \frac{a^2+b^2-2ab}{(a+b)^2} \\
  & = \frac{(b-a)^2}{(a+b)^2} \\
  & = (q-p)^2 \enspace .
\end{align*} 

Finally, let $p < q$ and let $z=pq$. This gives us the closed-form expression:
\begin{align*}
  \sum_{k=0}^\infty k {2k \choose k} (pq)^k &= \frac{2pq}{(1-4pq)^{3/2}} \\ 
  & = \frac{2pq}{((q-p)^2)^{3/2}} \\
  & = \frac{2pq}{(q-p)^3} \enspace .
\end{align*}

\end{document}